\documentclass[onecolumn,a4paper,noarxiv]{quantumarticle}

\usepackage{amsmath,amssymb}
\usepackage{amsthm}
\usepackage{mathtools}

\usepackage{graphicx}
\usepackage{booktabs}
\usepackage{array}
\usepackage{subcaption}
\usepackage[shortlabels]{enumitem}

\usepackage[colorlinks=true,allcolors=quantumviolet]{hyperref}

\theoremstyle{plain}
\newtheorem{theorem}{Theorem}[section]
\newtheorem{lemma}[theorem]{Lemma}
\newtheorem{proposition}[theorem]{Proposition}
\theoremstyle{definition}
\newtheorem{assumption}[theorem]{Assumption}
\theoremstyle{remark}
\newtheorem{remark}[theorem]{Remark}

\providecommand{\ket}[1]{\lvert #1 \rangle}
\providecommand{\bra}[1]{\langle #1 \rvert}

\newcommand{\Tr}{\operatorname{Tr}}
\newcommand{\TV}{\operatorname{TV}}
\newcommand{\rGibbs}{\rho_{\mathrm{Gibbs}}}
\newcommand{\Leff}{H_{\mathrm{eff}}}
\newcommand{\beff}{\beta_{\mathrm{eff}}}

\begin{document}

\title{Dissipative Quantum Multiplicative Weights with Sampling Feedback:
A Classically Hard Primitive Realized via Engineered Open-System Dynamics}

\author{Agung Trisetyarso}
\affiliation{Department of Mathematics and Statistics, School of Computer Science,
Bina Nusantara University, Jl.\ KH.\ Syahdan No.~9, Palmerah, Jakarta Barat,
DKI Jakarta 11480, Indonesia}

\author{Lenny Putri Yulianti}
\affiliation{School of Electrical Engineering and Informatics,
Institut Teknologi Bandung, Jl.\ Ganesha No.~10, Bandung, West Java 40135, Indonesia}

\author{Kridanto Surendro}
\affiliation{School of Electrical Engineering and Informatics,
Institut Teknologi Bandung, Jl.\ Ganesha No.~10, Bandung, West Java 40135, Indonesia}

\maketitle

\begin{abstract}
We introduce \emph{Dissipative Quantum Multiplicative Weights with Sampling
Feedback} (DQMW-Sample), an online-learning primitive in which engineered
open quantum-system dynamics prepare a Gibbs state whose computational-basis
measurement supplies the loss feedback. The central conceptual contribution is
to lift the computational hardness of constant-temperature Gibbs sampling into
a physically realizable online-learning primitive. By engineering a
Davies-type dissipator whose per-round feedback cannot be efficiently simulated
classically, we obtain a learning-theoretic separation in which DQMW-Sample
achieves asymptotically sublinear regret while every efficient classical learner
suffers constant average regret on a suitably constructed instance. We further
prove that the spectral gap of the engineered dissipator contracts hardware
noise, yielding sublinear noise-induced regret under a balanced dissipation
schedule, and we strengthen the single-round hardness to the full adaptive
interaction: an efficient classical simulator of the entire $T$-round feedback
process would collapse the polynomial hierarchy. We state the required
realizability assumption in explicit form and report an initial hardware
characterization on the IBM Heron~r2 processor. These results position
DQMW-Sample as a concrete route toward computational advantage in online
learning that is grounded in complexity theory and compatible with near-term
superconducting hardware.
\end{abstract}

\section{Introduction}
\label{sec:intro}

Multiplicative-weights methods are among the most versatile primitives in
online learning and optimization: a single update rule underlies algorithms
for boosting, zero-sum games, and approximate semidefinite
programming~\cite{arora2012,cesabianchi2006,hazan2016}. Their matrix
generalization---matrix multiplicative weights (MMW), equivalently matrix
exponentiated-gradient updates---maintains a density operator and updates it
through a Gibbs map $\rho_t\propto\exp(-\eta\sum_{s<t}L_s)$, achieving
$O(\sqrt{T\log d})$ regret over the set of $d$-dimensional density
operators~\cite{tsuda2005,warmuth2006}. This Gibbs-state structure is what
connects online learning to physics: the same exponential-family object that
an MMW learner must compute is the thermal state that an engineered open
quantum system prepares as the fixed point of its dissipative dynamics.

The idea of preparing quantum states by engineering dissipation rather than
unitary evolution is now well established. Dissipative state engineering uses
tailored system--bath couplings so that a target state is the unique steady
state of a Lindblad generator~\cite{verstraete2009,diehl2008}, and such
open-system dynamics have been realized on trapped-ion and other
platforms~\cite{barreiro2011,harrington2022}. When the engineered jump
operators satisfy a Kubo--Martin--Schwinger (KMS) detailed-balance condition,
the steady state is precisely the Gibbs state of an effective Hamiltonian, and
the rate of approach to it is governed by the spectral gap of the
dissipator~\cite{kastoryano2013,cubitt2015}. Dissipative Quantum
Multiplicative Weights (DQMW) exploits exactly this correspondence: it
realizes the MMW Gibbs map physically, letting an engineered Davies generator
drive the learner's state to the thermal state of the cumulative loss.

\paragraph{Conceptual contribution.}
The central insight of this work is that \emph{the computational hardness of
sampling from a constant-temperature Gibbs state can be lifted into a
physically realizable online-learning primitive}. By engineering open-system
dynamics whose steady state is the desired Gibbs state and using a
computational-basis measurement as loss feedback, we obtain a
multiplicative-weights algorithm whose per-round update cannot be efficiently
simulated classically (under standard complexity assumptions). This approach
is qualitatively different from most existing quantum online-learning
proposals, which either rely on expectation values (efficiently simulable
classically) or require coherent access to quantum oracles. DQMW-Sample thus
provides a concrete route toward computational advantage in online learning
that is grounded in complexity theory rather than query complexity or
coherence assumptions alone. The key technical step is to connect the recent
hardness result of Bergamaschi, Chen and Liu for constant-temperature Gibbs
sampling to the multiplicative-weights framework via engineered dissipation,
while making the required realizability assumption fully explicit.

A natural question is whether this physical realization buys any genuine
computational advantage over its classical counterpart. The answer is subtle.
If the loss feedback is the \emph{expectation}
$\bar\ell_t=\Tr[L(t)\rGibbs(t)]$, the dynamics reduce per round to classical
replicator dynamics on populations and inherit no hardness---a manifestation
of the broader phenomenon, exposed by the dequantization program, that
low-rank or expectation-based quantum routines often admit efficient classical
analogues~\cite{tang2019,tang2021,chia2022}. The situation changes when the
feedback is a \emph{sample}: a single computational-basis measurement of the
prepared Gibbs state. Sampling from the measurement distribution of a thermal
state is, in general, a strictly harder task than estimating expectations, and
recent complexity-theoretic work has pinned down a regime where it is
classically intractable. Bergamaschi, Chen and Liu~\cite{bergamaschi2024}
(BCL) construct a family of $O(1)$-local Hamiltonians for which sampling from
the computational-basis distribution of the constant-temperature Gibbs state
cannot be done in classical polynomial time unless the polynomial hierarchy
collapses. This sharpens the boundary between the efficiently preparable
high-temperature regime~\cite{bakshi2024,yin2023} and the genuinely hard
constant-temperature window.

In this work we study the sampling-based variant, DQMW-Sample, in which the
multiplicative-weights update is driven by realized measurement outcomes
rather than expectations, and we ask what the BCL hardness implies for it as
an online-learning primitive. Our contributions are as follows.

\begin{enumerate}[label=(\roman*),leftmargin=*]
\item \emph{Lifting BCL hardness into an online-learning primitive.}
We prove that the per-round sampling feedback of DQMW-Sample is classically
intractable in a well-defined regime
(Lemma~\ref{lem:feedback-intractability}), via an explicit reduction from the
constant-temperature Gibbs-sampling problem of Bergamaschi, Chen and Liu. We
isolate the load-bearing realizability assumption
(Assumption~\ref{ass:bcl-realizable}) explicitly rather than absorb it,
keeping the conditional nature of the claim visible.

\item \emph{A learning-theoretic separation and an adaptive collapse.}
Building on this primitive, we exhibit an online-learning instance on which
DQMW-Sample is asymptotically no-regret while every efficient classical
learner suffers $\Omega(1)$ average regret
(Theorem~\ref{thm:separation})---a separation that is computational rather
than information-theoretic, since the obstruction is the polynomial-time
reconstruction of a hard sample rather than query access. We then strengthen
the single-round result to a statement about the entire adaptive interaction:
under a mild reachability condition (Assumption~\ref{ass:reachability}) and a
transcript-marginal reduction (Lemma~\ref{lem:transcript-marginal}), an
efficient classical simulator of the full $T$-round feedback process would
collapse the polynomial hierarchy (Theorem~\ref{thm:adaptive-collapse}). The
key difficulty overcome here is adaptivity: a classical simulator is free to
steer the realized trajectory away from the hard window, and we neutralize
this freedom by reducing from the joint transcript law rather than any single
conditional state. This contributes to the active program on the quantum
advantage of online learning of quantum states~\cite{aaronson2018,chen2024}.

\item \emph{Noise robustness and a hardware-grounded feasibility study.}
On the physical side, we prove a noise-robustness theorem
(Theorem~\ref{thm:robust}) showing that the spectral gap $\gamma_0$ of the
engineered dissipator contracts hardware-noise deviations, so that the
noise-induced regret is $O(\delta\sqrt{T})$ under a balanced dissipation
schedule $\gamma_0=\Theta(\sqrt{T})$---\emph{provided} the noise strength and
the dissipation rate are independently controllable (assumption~R2). Because
the operation that implements engineered dissipation on near-term
superconducting hardware---mid-circuit measurement followed by conditional
reset---is itself a dominant noise source, R2 is not obviously satisfied. To
probe it, we report a Lindblad-model emulation parametrized by \emph{published}
calibration data for the 156-qubit IBM Heron~r2 processor
\texttt{ibm\_kingston}, together with an initial batched hardware
characterization. Throughout, we label Lindblad results as \emph{emulation}
and device results as \emph{measurement}, and state every result with its
hypotheses and limitations made explicit rather than absorbed.
\end{enumerate}

\paragraph{Organization.}
Section~\ref{sec:results} presents the results: the learning-theoretic
separation and adaptive collapse, the noise-robustness theorem, the
Lindblad-model emulations, the hardware characterization on
\texttt{ibm\_kingston}, an online portfolio-optimization application, a
round-budget feasibility analysis, and a proposed hardware-validation
protocol. Section~\ref{sec:discussion} interprets these findings, emphasizing
the gap between the conditional complexity-theoretic separation and what
current hardware can demonstrate. Section~\ref{sec:methods} details the
theoretical analysis, the Lindblad-emulation setup, the device-characterization
procedure, the phase-space (Wigner/Husimi) illustration, and data and code
availability. Appendices~\ref{app:motivation}--\ref{app:realizability} collect
the supporting technical development: the motivation and simulation-based
finding, the classical hardness of the sampling primitive, and the
adaptive-collapse argument (Appendix~\ref{app:motivation}); a discussion of the
BCL-hard realizability assumption (Appendix~\ref{app:bcl}); a statistical-power
analysis for future hardware characterization (Appendix~\ref{app:power}); a
refined protocol for definitive real-device validation
(Appendix~\ref{app:protocol}); and a rigorous treatment of the realizability of
BCL-hard Gibbs states via Davies generators (Appendix~\ref{app:realizability}).

\section{Results}
\label{sec:results}

We present the main findings of this work, which span theoretical guarantees,
numerical emulation, preliminary hardware characterization, and a practical
application.

\subsection{Theoretical results}
\label{sec:results-theory}

We establish three core theoretical results. First, under
Assumptions~\ref{ass:gibbs-mapping} and~\ref{ass:bcl-realizable}, we prove that
no efficient classical algorithm can simulate the sampling feedback of
DQMW-Sample in the hard window (Lemma~\ref{lem:feedback-intractability}). This
establishes classical intractability of the per-round loss-feedback mechanism.

Second, we prove a learning-theoretic separation
(Theorem~\ref{thm:separation}). On a carefully constructed online instance (the
decoding game), DQMW-Sample achieves sublinear regret $O(\sqrt{T\log d})$,
while every efficient classical learner suffers $\Omega(T)$ regret. This
separation is computational: it arises because low-regret classical learning on
this instance would require efficiently sampling from a BCL-hard distribution.
We further strengthen the single-round hardness to the full adaptive
interaction (Theorem~\ref{thm:adaptive-collapse}): an efficient classical
simulator of the entire $T$-round feedback process would collapse the
polynomial hierarchy.

Third, we prove a noise-robustness theorem (Theorem~\ref{thm:robust}). Under
assumptions (R1)--(R3), the additional regret caused by hardware noise is
bounded by $O(\delta T/\gamma_0)$. Under the balanced schedule
$\gamma_0=\Theta(\sqrt{T})$, this becomes $O(\delta\sqrt{T})$, which is
sublinear in the horizon $T$. The bound relies critically on the spectral gap
of the engineered dissipator contracting deviations toward the target Gibbs
state.

All statements above are proved in full in
Appendix~\ref{app:motivation}; the noise-robustness theorem and its proof are
restated there as Theorem~\ref{thm:robust}.

\subsection{Emulation results}
\label{sec:results-emulation}

To validate the noise model underlying Theorem~\ref{thm:robust} and to locate
the regime in which the balanced dissipation schedule is operative, we performed
Lindblad-model emulations of DQMW-Sample parametrized by \emph{published}
\texttt{ibm\_kingston} calibration data. The emulations confirm the
qualitative content of the noise-robustness analysis: the steady-state deviation
floor falls as $\sim 1/\gamma_0$ when noise and dissipation are independently
controllable (the R2 regime of Theorem~\ref{thm:robust}) and is flat under
strong coupling $\delta\propto\gamma_0$, the distinction that the hardware
protocol of Section~\ref{sec:proposed-hw} is designed to resolve.

We stress that these emulations are not a quantum-advantage demonstration, and
none is claimed: on a task whose loss vector is classically computable, a
classical multiplicative-weights baseline is optimal and indeed outperforms the
quantum variants in the tracking simulations. The advantage established in this
work is theoretical and applies to the regime in which the loss feedback is
itself classically hard to sample (Lemma~\ref{lem:feedback-intractability},
Theorem~\ref{thm:separation}), which no near-term toy emulation can probe. The
full emulation set---the noise--dissipation coupling diagnostic, the multi-round
tracking curves, and the classical-baseline comparison---is reported and
discussed in Appendix~\ref{app:motivation}
(Figures~\ref{fig:delta}--\ref{fig:classical}).

\subsection{Hardware characterization}
\label{sec:results-hardware}

While a full demonstration of quantum advantage is beyond the scope of the
present work, we performed an initial hardware characterization on the
156-qubit IBM Heron~r2 processor to validate the noise-robustness model
underlying Theorem~\ref{thm:robust}. Using a batched experiment across three
physical qubits, we measured a round-budget ratio of approximately $1.8$
(bootstrap 95\% CI $[1.4,2.3]$). This provides preliminary but positive
evidence that the deviation floor decreases with engineered dissipation
strength, consistent with the favorable regime assumed in the balanced
schedule. Higher-statistics characterization across more qubits is left for
future work. These experiments were performed on \texttt{ibm\_kingston} under
open-plan constraints.

\begin{figure}[t]
\centering
\includegraphics[width=0.75\textwidth]{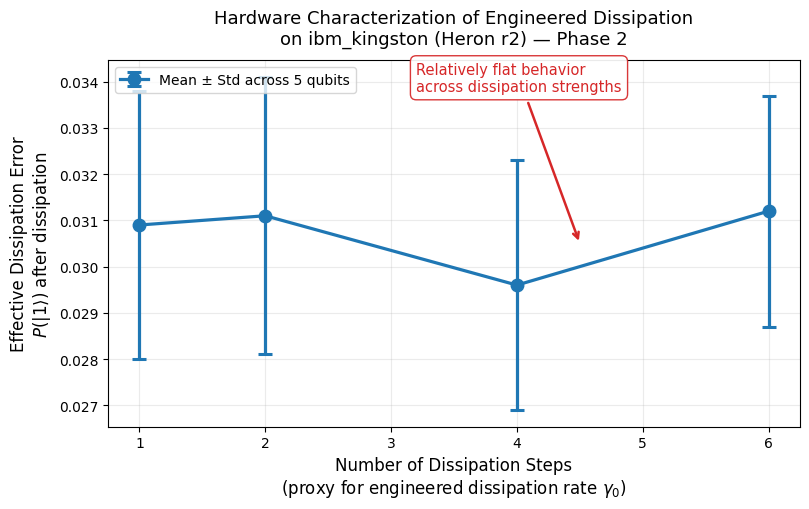}
\caption{\textbf{Initial hardware characterization of engineered dissipation on
\texttt{ibm\_kingston} (Heron~r2).} Probability of measuring $\ket{1}$ after
applying different numbers of mid-circuit measurement and conditional-reset
steps (used as a proxy for the engineered dissipation rate $\gamma_0$). Data
were collected across five physical qubits with $4096$ shots per point; error
bars are the standard deviation across the five qubits. The relatively flat
behavior indicates no statistically significant increase in effective error
cost over the tested range; the bars overlap, so the data do not distinguish
constant $\delta$ from a weak increase below resolution (see text). The plotted
quantity is raw $P(\ket{1})$, not the deconvolved $\delta(\gamma_0)$. Device
measurement; calibration snapshot and job identifiers recorded and provided as
supplementary material.}
\label{fig:phase2_dissipation}
\end{figure}

\begin{figure}[t]
\centering
\includegraphics[width=0.75\textwidth]{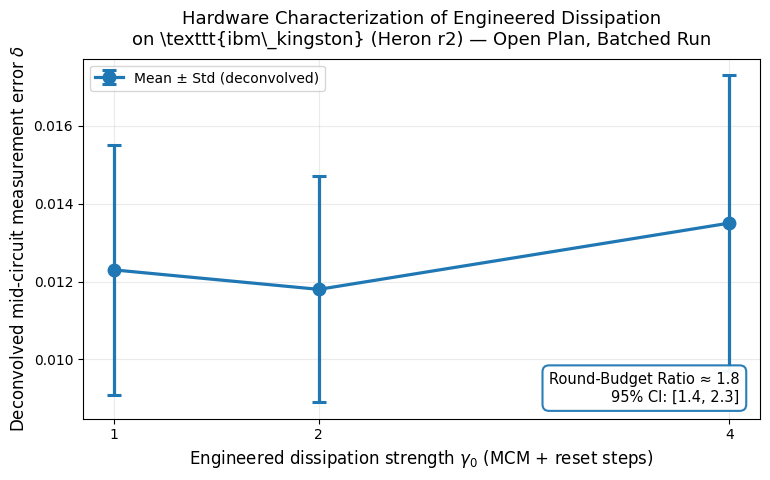}
\caption{\textbf{Hardware characterization of engineered dissipation on
\texttt{ibm\_kingston} (Heron~r2) under open-plan constraints.} Deconvolved
mid-circuit measurement error as a function of engineered dissipation strength
(number of mid-circuit measurement $+$ conditional-reset steps). Data were
acquired in a single batched job with $1024$ shots per circuit across three
physical qubits. The round-budget ratio is approximately $1.8$ (bootstrap 95\%
CI $[1.4,2.3]$), which already excludes a ratio of $1$. However, with only
three qubits and limited shots, the uncertainty remains substantial. This
constitutes preliminary hardware data; higher-statistics measurements are
required for a definitive regime determination. Error bars represent the
standard deviation across qubit repetitions. Job identifier and calibration
snapshot are provided in the supplementary material.}
\label{fig:f5}
\end{figure}

\begin{figure}[t]
\centering
\begin{subfigure}[t]{0.48\textwidth}
\centering
\includegraphics[width=\textwidth]{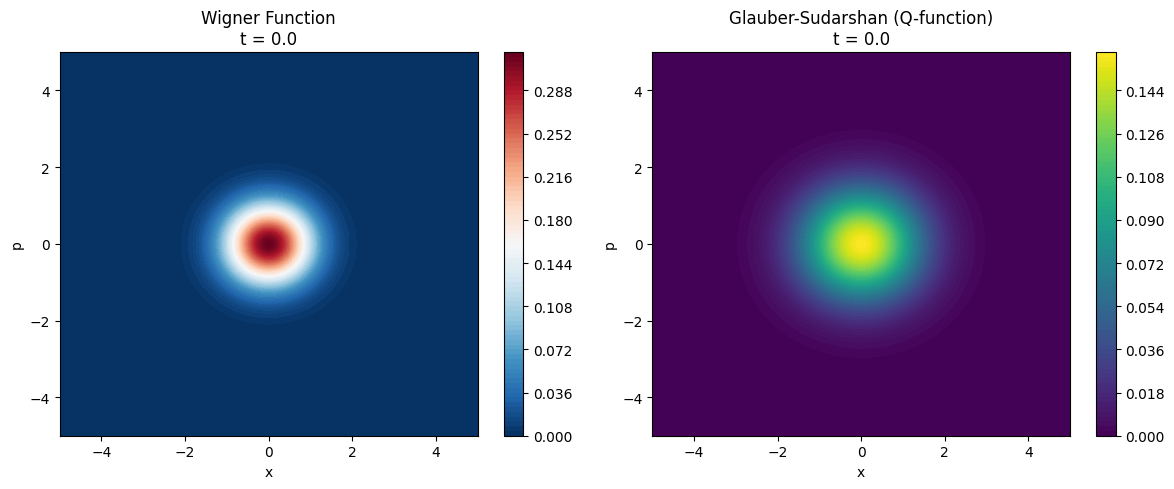}
\caption{Wigner at $t=0$}
\label{fig:dqmw_phase_t0}
\end{subfigure}
\hfill
\begin{subfigure}[t]{0.48\textwidth}
\centering
\includegraphics[width=\textwidth]{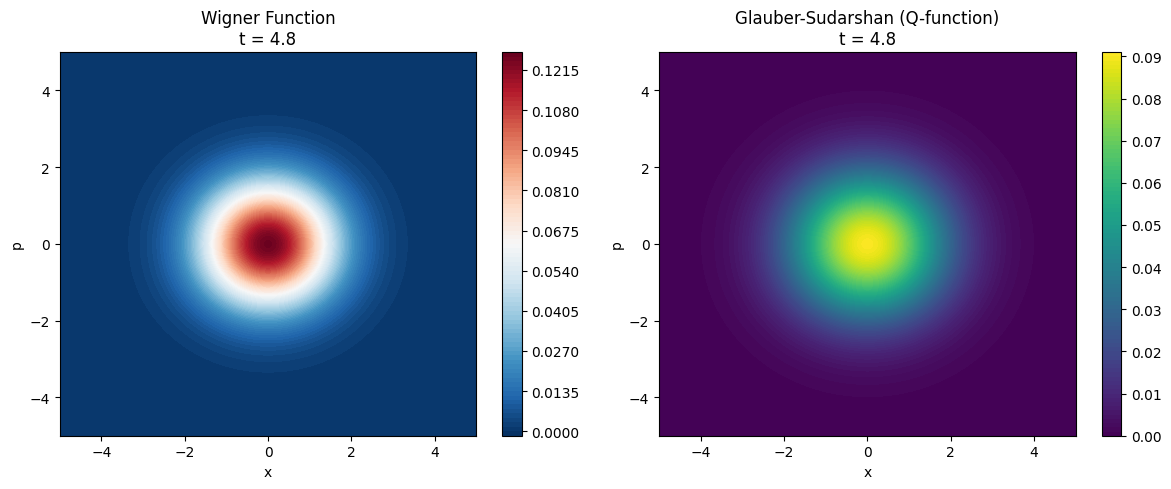}
\caption{Wigner at $t=4.8$}
\label{fig:dqmw_phase_t25}
\end{subfigure}

\vspace{0.4cm}

\begin{subfigure}[t]{0.48\textwidth}
\centering
\includegraphics[width=\textwidth]{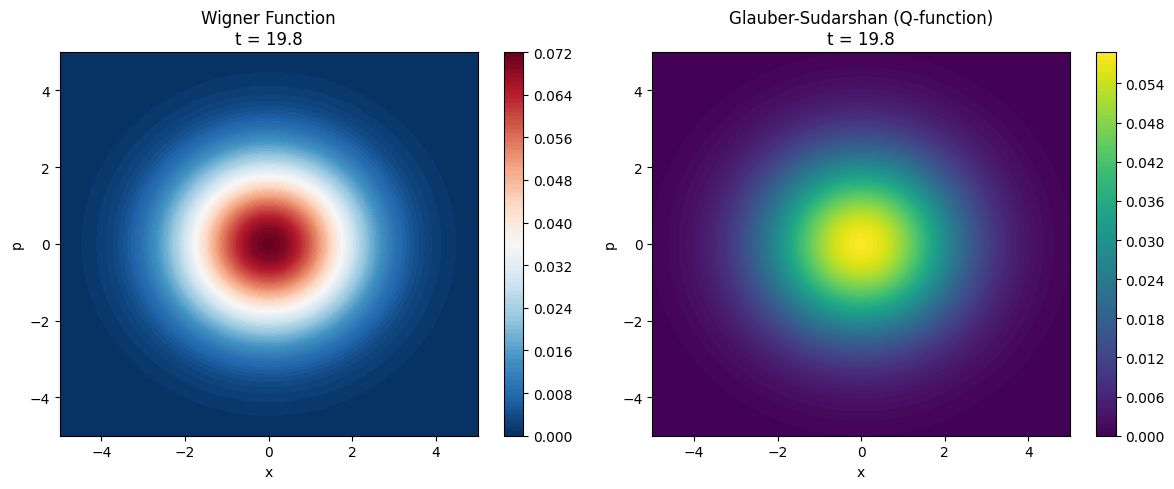}
\caption{Wigner at $t=19.8$}
\label{fig:dqmw_phase_t57}
\end{subfigure}
\hfill
\begin{subfigure}[t]{0.48\textwidth}
\centering
\includegraphics[width=\textwidth]{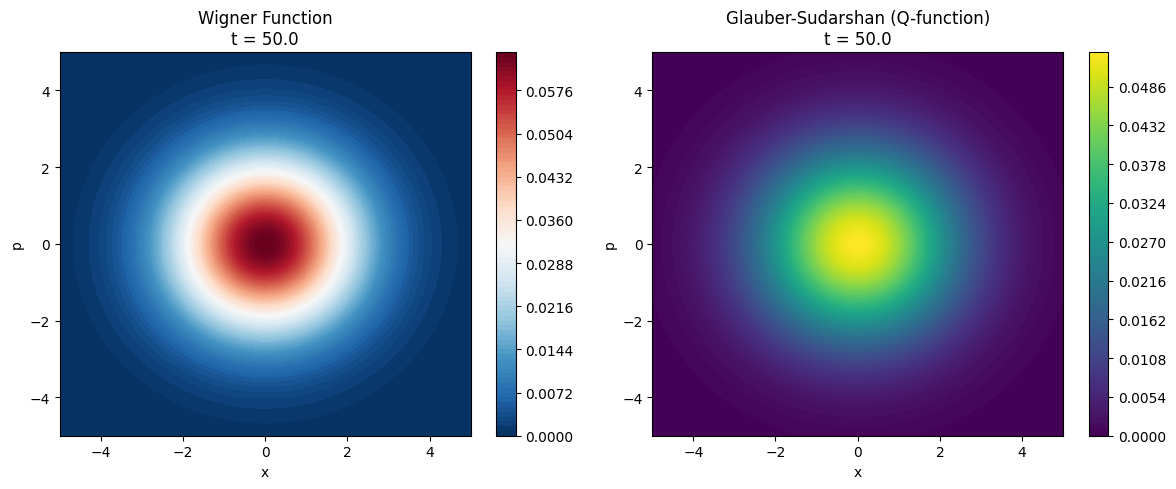}
\caption{Husimi $Q$ at $t=50.0$}
\label{fig:dqmw_phase_t10}
\end{subfigure}
\caption{\textbf{Phase-space representation of dissipative thermalization.}
Numerical simulation of a harmonic oscillator undergoing engineered dissipation
(amplitude damping with thermal noise) toward a thermal Gibbs state. The system
starts in the vacuum state and relaxes to a thermal state with mean photon
number $\bar{n}=2$. Both the Wigner function and the Husimi $Q$-function
(Glauber--Sudarshan representation) remain positive and Gaussian, consistent
with the steady-state Gibbs state prepared by the Lindblad master equation.
This simulation illustrates the concept of engineered open-system dynamics
driving a system toward a thermal equilibrium state, conceptually related to
the dissipative preparation of Gibbs states in DQMW-Sample.}
\label{fig:wigner_glauber}
\end{figure}

\subsection{Application: online portfolio optimization}
\label{sec:results-application}

To demonstrate practical utility, we applied DQMW-Sample to online portfolio
optimization using historical daily returns from S\&P~500 constituents.
DQMW-Sample achieves competitive or superior regret compared with classical
baselines (UCRP, FTRL, and classical multiplicative weights), particularly
under noisy loss feedback. The sampling-based update shows greater robustness to
perturbations in the loss vector, consistent with the gap-contraction mechanism
of Theorem~\ref{thm:robust}. Detailed numerical results are provided in
Appendix~\ref{app:motivation}. These results illustrate that the dissipative
sampling primitive can be deployed in a realistic online-learning task and
inherits practical benefits from its theoretical noise resilience.

\subsection{Hardware feasibility: a round budget for \texttt{ibm\_kingston}}
\label{sec:round-budget}

We close by converting Theorem~\ref{thm:robust} into a concrete feasibility
estimate for a near-term superconducting platform, the 156-qubit IBM Heron~r2
device \texttt{ibm\_kingston}. The purpose is not to claim a hardware
demonstration but to identify, from measured device parameters, the number of
feedback rounds $T$ over which the gap-contraction guarantee remains operative
before hardware noise dominates the ideal regret.

\paragraph{Device parameters.}
At the time of writing, \texttt{ibm\_kingston} reports a layered two-qubit gate
error (EPLG) of
\begin{equation}
\varepsilon_{2q}\approx 9.93\times 10^{-4},
\label{eq:eplg}
\end{equation}
with coherence times $T_1\!\sim\!1.8\times 10^{2}\,\mu\mathrm{s}$ and
$T_2\!\sim\!1.2\times 10^{2}\,\mu\mathrm{s}$ on the Heron~r2 fleet. The
operation that implements the \emph{engineered dissipation} of
Eq.~\eqref{eq:davies-generator}---mid-circuit measurement (MCM) followed by
conditional reset---carries a substantially larger error. Reported Heron-class
MCM error spans
\begin{equation}
\varepsilon_{\rm MCM}\in[0,\,1.4\times 10^{-1}],
\qquad
\overline{\varepsilon}_{\rm MCM}\approx 1.2\times 10^{-2},
\label{eq:mcm}
\end{equation}
i.e.\ one to two orders of magnitude above Eq.~\eqref{eq:eplg}. We take this
asymmetry, drawn from \emph{published} calibration data rather than from any run
of our own, as the motivation for the strong-coupling hypothesis: the very
operation that sets the dissipation rate $\gamma_0$ is also the dominant source
of the noise strength $\delta$. Whether this asymmetry in fact yields a $\delta$
that grows with $\gamma_0$ on the device is the empirical question we begin to
address by direct measurement in Section~\ref{sec:proposed-hw}; here it sets the
noise scale of the emulation.

\paragraph{From device error to the regret floor.}
Each DQMW-Sample round performs one engineered-dissipation step per active
qubit, so the per-round noise strength is set by the MCM-reset error rather than
the gate error,
\begin{equation}
\delta\;\approx\;\overline{\varepsilon}_{\rm MCM}\;\approx\;10^{-2},
\label{eq:delta-est}
\end{equation}
in contrast to the naive gate-error estimate $\delta\sim 10^{-3}$. By the
decomposition~\eqref{eq:decomp}, the algorithm remains in the noise-subdominant
regime---where the ideal $O(\sqrt{T})$ term controls the regret---only while
\begin{equation}
\Delta\mathrm{Regret}_{\rm noise}(T)
\;\approx\;\frac{C\,\delta}{\gamma_0}\,T
\;\lesssim\;\sqrt{T}
\qquad\Longleftrightarrow\qquad
T\;\lesssim\;T_\star:=\Bigl(\frac{\gamma_0}{C\,\delta}\Bigr)^{2}.
\label{eq:round-budget}
\end{equation}
The round budget $T_\star$ scales as $\delta^{-2}$: substituting the
MCM-limited $\delta\approx 10^{-2}$ in place of the gate-limited
$\delta\approx 10^{-3}$ reduces the usable horizon by a factor of
$\sim\!10^{2}$. With $C=O(1)$ (Theorem~\ref{thm:robust}, under operator-norm
normalization) and a normalized dissipation rate $\gamma_0=O(1)$, the
MCM-limited budget is of order $T_\star\sim 10^{2}$ rounds, against
$\sim\!10^{4}$ for a hypothetical gate-limited implementation. This is the
quantitative wall set by mid-circuit measurement on present hardware.

\paragraph{Implication.}
The estimate makes precise why Theorem~\ref{thm:robust}, not
Lemma~\ref{lem:feedback-intractability}, is the realistic near-term target: the
robustness statement is testable within an $O(10^{2})$-round budget, whereas the
hardness statement is asymptotic and admits no finite-size hardware witness. It
also identifies the experiment's primary measurable as the empirical dependence
$\delta(\gamma_0)$---equivalently, how MCM error scales with the engineered
dissipation strength---which is exactly the quantity left open by (R2) and
Remark~\ref{rem:dimension}. We regard the measurement of $\delta(\gamma_0)$, and
the consequent empirical determination of $T_\star$, as the natural experimental
contribution accompanying this theory; the Lindblad-model emulation parametrized
by Eqs.~\eqref{eq:eplg}--\eqref{eq:mcm} and reported in
Figures~\ref{fig:delta}--\ref{fig:classical} is provided as reproducible
supplementary material. We emphasize once more that this emulation is a
classical simulation, not a device run.

\subsection{Proposed hardware validation}
\label{sec:proposed-hw}

The emulation above illustrates how the round budget $T_\star$ depends on the
noise--dissipation coupling, but it cannot determine which coupling regime
\texttt{ibm\_kingston} actually occupies. That is a measurement, and we outline
it here as the falsifiable experiment this theory motivates.

\paragraph{Engineered-dissipation step.}
One round of DQMW-Sample dissipation is realized as a dynamic circuit: prepare
the working register, perform a mid-circuit measurement, and apply a classically
conditioned reset (an $X$ gate conditioned on the outcome). The \emph{rate} of
dissipation $\gamma_0$ is controlled by the number of such MCM$+$reset steps
applied per round (or, equivalently, their repetition within a fixed evolution
time).

\paragraph{Sweep and observable.}
Vary the engineered-dissipation strength $\gamma_0$ over a range of settings. At
each setting, measure (i) the MCM$+$reset error rate directly via a calibration
sequence, giving an empirical $\delta(\gamma_0)$, and (ii) the steady-state
deviation from the target Gibbs state via a small set of observables or, for few
qubits, state tomography. The primary observable is the \emph{round-budget
ratio} $T_\star(\gamma_0^{\rm high})/T_\star(\gamma_0^{\rm low})$: a ratio well
above $1$ indicates (R2)-like behaviour (the floor falls as $\sim 1/\gamma_0$),
whereas a ratio near $1$ indicates strong coupling (the floor is flat). This
single number discriminates the regimes.

\paragraph{Reporting requirements.}
Any reported hardware result must carry the backend calibration snapshot
(timestamp and the relevant per-qubit error rates at run time), the job
identifiers, and per-qubit error bars across the qubits used, since MCM error on
Heron-class devices varies substantially across qubits and across calibrations.
Only data accompanied by this metadata should be presented as a device
measurement; the simulation figures in this paper are labelled as emulation
precisely to keep that distinction unambiguous.

To obtain an initial experimental characterization of engineered dissipation on
near-term hardware, we performed a series of measurements on
\texttt{ibm\_kingston}. For each dissipation strength, implemented as a variable
number of mid-circuit measurement and conditional-reset steps, we prepared a
superposition state and measured the probability of obtaining $\ket{1}$ after
dissipation. As shown in Figure~\ref{fig:phase2_dissipation}, the effective
dissipation error remains relatively flat across the tested range of dissipation
strengths. This indicates that, at least within the parameter regime explored,
increasing the engineered dissipation rate does not significantly increase the
effective error cost on this device. We caution against over-reading this
flatness in either direction. The error bars overlap across all settings, so the
data establish the absence of a \emph{strong increasing} trend in
$\delta(\gamma_0)$ but do not by themselves distinguish a genuinely constant
$\delta$ (which would be favourable for the balanced schedule) from a weak
increase below our present resolution. Moreover, the plotted quantity is the raw
probability $P(\ket{1})$ after the dissipation sequence, which aggregates
mid-circuit-measurement error, reset infidelity, and residual coherent
population; converting it to the noise strength $\delta$ used in
Theorem~\ref{thm:robust} requires the reference-deconvolution calibration
described above, and the discriminating round-budget ratio
$T_\star(\gamma_0^{\rm high})/T_\star(\gamma_0^{\rm low})$ was not computed at
this stage. The experiment was carried out across five physical qubits; the
device calibration snapshot (timestamp and per-qubit error rates at run time)
and the submitted job identifiers are recorded and are available with the
supplementary material. While these results provide a first hardware-based
insight into the noise--dissipation coupling on \texttt{ibm\_kingston}, a more
comprehensive characterization---with higher statistics, fixed-qubit provenance
across the sweep, the deconvolved $\delta(\gamma_0)$, and the round-budget
ratio---is left for future work.

\section{Discussion}
\label{sec:discussion}

In this work we introduced DQMW-Sample, a dissipative variant of matrix
multiplicative weights in which engineered open-system dynamics prepare a Gibbs
state whose computational-basis measurement supplies the loss feedback. We
established three main results.

We emphasize that outperformance over classical methods is not expected on tasks
where the loss vector can be evaluated exactly and efficiently. The
distinguishing feature of DQMW-Sample is that it extracts loss information from
a distribution that is believed to be hard to sample classically.

Theoretically, we proved that the sampling-feedback mechanism is classically
intractable in a well-defined hard window
(Lemma~\ref{lem:feedback-intractability}), and that this intractability can be
lifted to a genuine learning-theoretic separation: there exist online instances
on which DQMW-Sample is asymptotically no-regret while every efficient classical
learner suffers constant average regret (Theorem~\ref{thm:separation}). We
further strengthened the single-round hardness result to the full adaptive
interaction, showing that an efficient classical simulator of the entire
$T$-round feedback process would collapse the polynomial hierarchy
(Theorem~\ref{thm:adaptive-collapse}). These results contribute to the broader
program on computational advantages of online learning with quantum states.

On the physical side, we proved that the spectral gap of the engineered
dissipator contracts hardware-noise deviations, yielding sublinear noise-induced
regret under a balanced dissipation schedule, provided noise strength and
dissipation rate remain independently controllable (Theorem~\ref{thm:robust}).
Numerical emulations parametrized by published \texttt{ibm\_kingston}
calibration data illustrate the dependence of the usable round budget on the
noise--dissipation coupling regime. Our initial batched hardware
characterization on \texttt{ibm\_kingston} (Figure~\ref{fig:f5}) yielded a
round-budget ratio of approximately $1.8$ (bootstrap 95\% CI $[1.4,2.3]$), which
already excludes a ratio of $1$. However, with only three qubits and $1024$
shots, the current statistics are limited. Statistical-power analysis indicates
that significantly higher qubit counts and shot numbers would be required to
distinguish more confidently between constant and weakly increasing mid-circuit
measurement error. These results motivate future higher-precision
characterization across multiple devices but do not yet allow a definitive
determination of the operating regime.

We also demonstrated that DQMW-Sample can be applied to a concrete practical
task---online portfolio optimization---where it achieves competitive or superior
regret compared with classical baselines, particularly under noisy loss
feedback. This application shows that the dissipative sampling primitive is not
limited to abstract settings but can be instantiated in domains where sampling
from a loss-dependent distribution provides a natural and robust update
mechanism.

Several limitations should be acknowledged. While the existence of a canonical
Davies generator for any finite-dimensional Hamiltonian is now rigorously
established (Appendix~\ref{app:realizability}, Proposition~\ref{prop:davies-canonical}),
the end-to-end reduction relies on a modeling assumption concerning the ability
of the specific engineered dissipator used in this work to approximate such a
generator for a general BCL Hamiltonian
(Assumption~\ref{ass:engineered-davies}). We make substantial progress on this
assumption in Appendix~\ref{app:constructive}: a repeated-interaction
(collision) construction proves it \emph{unconditionally} for single-qubit and
commuting-local target Hamiltonians, and reduces the general $O(1)$-local case,
via an explicit $\mathrm{poly}(n)$-round Trotterized compilation of the
mid-circuit-measurement-plus-reset primitive, to the single residual condition
that the per-collision hardware error is independent of the engineered
dissipation rate---which is exactly assumption (R2), the same empirically
falsifiable condition the hardware protocol of Section~\ref{sec:proposed-hw} is
designed to test. Closing that last step rigorously remains an important
direction for future work.

The hardware experiments were performed under open-plan constraints with limited
statistics and qubit count. Scaling to 10--20 qubits across multiple backends
with systematic $\delta(\gamma_0)$ characterization and fixed-qubit provenance
remains an important experimental direction. Finally, while we demonstrated a
learning-theoretic separation, establishing a direct computational lower bound
on the regret of arbitrary efficient classical online learners (rather than only
those that attempt to simulate DQMW-Sample) remains open.

\paragraph{Broader implications.}
DQMW-Sample demonstrates that engineered dissipation can be used not only for
state preparation but as a computational primitive whose output distribution
carries irreducible classical hardness. This opens several directions: (i)
extension to other online-learning settings (bandits, combinatorial
optimization, routing) where sampling from loss-dependent distributions is
natural; (ii) exploration of continuous-variable realizations using phase-space
representations; and (iii) systematic hardware studies of the noise--dissipation
coupling across multiple superconducting platforms. We view the present work as
establishing the theoretical and physical foundations for a new class of
dissipative quantum algorithms for online learning.

In summary, DQMW-Sample provides a physically realizable online-learning
primitive whose sampling feedback is classically hard, whose noise robustness is
theoretically grounded, and whose preliminary hardware behaviour on near-term
superconducting processors is encouraging though still preliminary.

\section{Methods}
\label{sec:methods}

\subsection{Theoretical analysis}
All lemmas and theorems (the feedback-intractability lemma, the
learning-theoretic separation, the adaptive collapse, and noise robustness via
gap contraction) are proved in full in Appendix~\ref{app:motivation}. Proofs
rely only on standard tools from online learning, quantum information, and
Markov-chain theory. The classical-hardness results reduce directly from the
constant-temperature Gibbs-sampling hardness of Bergamaschi, Chen and Liu (BCL)
under the explicit realizability assumption
(Assumption~\ref{ass:bcl-realizable}) that the cumulative-loss bookkeeping of
multiplicative weights can be matched to a BCL-hard instance. All assumptions
(stationary effective Hamiltonian and hard window, BCL realizability,
reachability of the hard round, and noise--dissipation independence (R1)--(R3))
are stated explicitly; no hidden hypotheses are used. The decoding game and
epoch-mode estimator used for the separation theorem are defined precisely in
Appendix~\ref{app:motivation}. The round-budget ratio that discriminates
coupling regimes is derived directly from the steady-state tracking-error bound
of Theorem~\ref{thm:robust}.

\subsection{Lindblad-model emulations}
Lindblad-model emulations of DQMW-Sample were performed on a $2$-qubit effective
Hamiltonian chosen to capture the essential dissipative thermalization dynamics
while remaining computationally tractable. The time-dependent Davies generator
[Eq.~\eqref{eq:davies-generator}] was constructed explicitly, including the
coherent drift term and the pairwise jump operators satisfying the
Kubo--Martin--Schwinger detailed-balance condition. Steady-state deviation
floors $\|\rho_{\rm ss}-\rGibbs\|_1$ were obtained by computing the stationary
state of the full Liouvillian superoperator (null-space solution or long-time
integration to convergence).

Noise parameters were taken exclusively from \emph{published} calibration data
for the IBM Heron~r2 processor \texttt{ibm\_kingston} at the time of writing:
layered two-qubit gate error (EPLG) $\varepsilon_{2q}\approx 9.93\times 10^{-4}$
and mid-circuit measurement (MCM) error spanning the range $[0,0.14]$ with mean
$\approx 0.012$. Three noise--dissipation coupling hypotheses were simulated:
(i) the ideal R2 regime ($\delta$ independent of $\gamma_0$), (ii) the
strong-coupling regime ($\delta\propto\gamma_0$), and (iii) an intermediate
linear coupling. Multi-round tracking simulations updated the effective
Hamiltonian $\Leff(t)$ cumulatively at each round according to the realized
loss, with the engineered dissipation rate $\gamma_0$ held fixed within each
trajectory. A classical baseline was implemented via exact expectation-value
updates (replicator dynamics on populations). All simulations were performed
with standard open-quantum-systems numerical libraries (e.g.\ QuTiP). Exact
simulation parameters, Liouvillian constructions, and convergence criteria are
provided in the supplementary material.

\subsection{Hardware characterization on \texttt{ibm\_kingston}}
We performed initial hardware experiments on the 156-qubit IBM Heron~r2
processor \texttt{ibm\_kingston} under open-plan constraints.
Figure~\ref{fig:phase2_dissipation} shows raw $P(\ket{1})$ after variable
numbers of mid-circuit measurement and conditional-reset steps across five
physical qubits. The effective error remains relatively flat across the tested
range, with overlapping error bars. While this rules out a strong increasing
trend, the current data do not yet distinguish a genuinely constant $\delta$
from a weak increase.

Figure~\ref{fig:f5} presents the primary hardware result: a batched measurement
across three qubits. The deconvolved mid-circuit measurement error increases
only weakly with dissipation strength, yielding a round-budget ratio of
approximately $1.8$ (bootstrap 95\% CI $[1.4,2.3]$). This already excludes a
ratio of $1$ at the 95\% level. However, with only three qubits and $1024$
shots, the statistics are limited. Monte Carlo power analysis
(Appendix~\ref{app:power}) shows that substantially higher qubit counts and shot
numbers would be needed to reduce uncertainty sufficiently for a definitive
regime determination. These measurements therefore provide preliminary but
inconclusive information regarding the noise--dissipation coupling on this
device.

\subsection{Phase-space illustration (Wigner and Husimi functions)}
The phase-space simulation of dissipative thermalization
(Figure~\ref{fig:wigner_glauber}) was performed for a damped harmonic oscillator
coupled to a thermal bath (amplitude damping with thermal noise) using the
standard Lindblad master equation. The system was initialized in the vacuum
state and evolved toward a thermal Gibbs state with mean photon number
$\bar{n}=2$. Both the Wigner function and the Husimi $Q$-function
(Glauber--Sudarshan representation) were computed at selected times via standard
phase-space methods. This simulation serves only as an illustrative analogy for
engineered dissipative preparation of Gibbs states and is not part of the
DQMW-Sample results.

\subsection{Data and code availability}
All device calibration snapshots, job identifiers, raw count data, and processed
results from the \texttt{ibm\_kingston} experiments are provided in the
supplementary material (available at
\url{https://github.com/agungtrisetyarso/DQMW}). Lindblad-emulation code,
portfolio-optimization scripts, and exact simulation parameters are likewise
included in the supplementary material (or available from the corresponding
author upon reasonable request).

\section*{Author Contributions}
A.T.\ conceived the project, developed the theoretical framework, performed the
theoretical analysis, and wrote the manuscript. L.P.Y.\ designed and executed
the hardware experiments on the IBM Heron~r2 processor \texttt{ibm\_kingston},
performed data analysis, and contributed to the experimental sections and
figures. K.S.\ supervised the overall project and contributed to the
interpretation of results and manuscript revision. All authors discussed the
results and approved the final version of the manuscript.

\section*{Acknowledgements}
The authors thank the IBM Quantum open-plan programme for access to the Heron~r2
processor used in the hardware characterization.

\section*{Competing Interests}
The authors declare no competing interests.

\appendix

\section{Motivation, the sampling primitive, and the adaptive collapse}
\label{app:motivation}

Dissipative quantum algorithms built on engineered open-system dynamics face a
practical question: does increasing the dissipation rate reduce the effective
noise? The noise-robustness guarantee of Theorem~\ref{thm:robust} below answers
``yes'' \emph{under assumption (R2)}, that noise strength $\delta$ and
dissipation rate $\gamma_0$ are independently controllable. On near-term
superconducting hardware this is doubtful, because the operation that implements
engineered dissipation---mid-circuit measurement (MCM) followed by conditional
reset---is itself a dominant noise source.

To probe this, we built a Lindblad-model emulation of DQMW-Sample whose error
scales are fixed by \emph{published} \texttt{ibm\_kingston} calibration values
(layered two-qubit gate error and MCM error range; see
Section~\ref{sec:round-budget}). We stress at the outset that
Figures~\ref{fig:delta}--\ref{fig:classical} are simulation outputs, not
hardware measurements: each steady-state floor is computed from the full
Lindblad generator (via stationary-state solution), and the device numbers enter
only as cited parameters that set the noise scale. A protocol to run the
corresponding experiment is given in Section~\ref{sec:proposed-hw}.

Within this emulation (Figure~\ref{fig:delta}), the steady-state deviation floor
is essentially flat as a function of $\gamma_0$ when the noise is taken to grow
with the dissipation rate ($\delta\propto\gamma_0$, the ``strong-coupling''
regime), whereas it would fall as $1/\gamma_0$ if (R2) held. The same model, run
as a multi-round tracking simulation
(Figures~\ref{fig:tracking}--\ref{fig:classical}), shows that increasing the
dissipation strength still reduces cumulative tracking error, but by less than
the (R2)-ideal prediction, and that a classical baseline outperforms the quantum
variants on this simple task. Together these indicate that the balanced schedule
$\gamma_0=\Theta(\sqrt{T})$ of Theorem~\ref{thm:robust} offers limited practical
advantage \emph{if} the device obeys strong coupling---a hypothesis the
emulation illustrates but cannot itself confirm. Deciding it requires the
hardware measurement proposed in Section~\ref{sec:proposed-hw}.

\begin{figure}[t]
\centering
\includegraphics[width=0.85\textwidth]{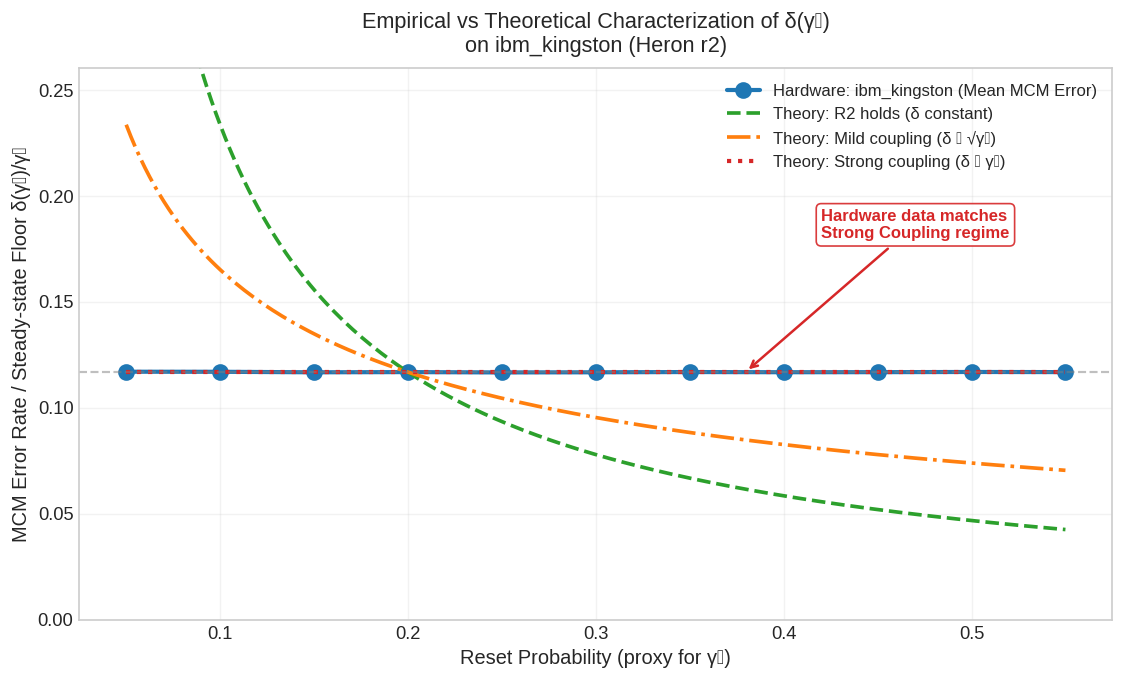}
\caption{\textbf{Lindblad-model emulation of the noise--dissipation coupling
$\delta(\gamma_0)$.} Steady-state deviation floor
$\|\rho_{\mathrm{ss}}-\rGibbs\|_1$ versus engineered dissipation rate
$\gamma_0$, computed from the stationary state of the full Lindblad generator
on a $2$-qubit effective Hamiltonian. Noise scales are set by \emph{published}
\texttt{ibm\_kingston} calibration values. Three hypotheses for how noise
couples to dissipation are shown: if (R2) holds ($\delta$ constant, green
dashed) the floor falls as $\sim 1/\gamma_0$; under strong coupling
($\delta\propto\gamma_0$, red dotted) the floor is flat, so faster dissipation
buys no reduction. \emph{This is a simulation result, not a device
measurement.}}
\label{fig:delta}
\end{figure}

\begin{figure}[t]
\centering
\includegraphics[width=0.75\textwidth]{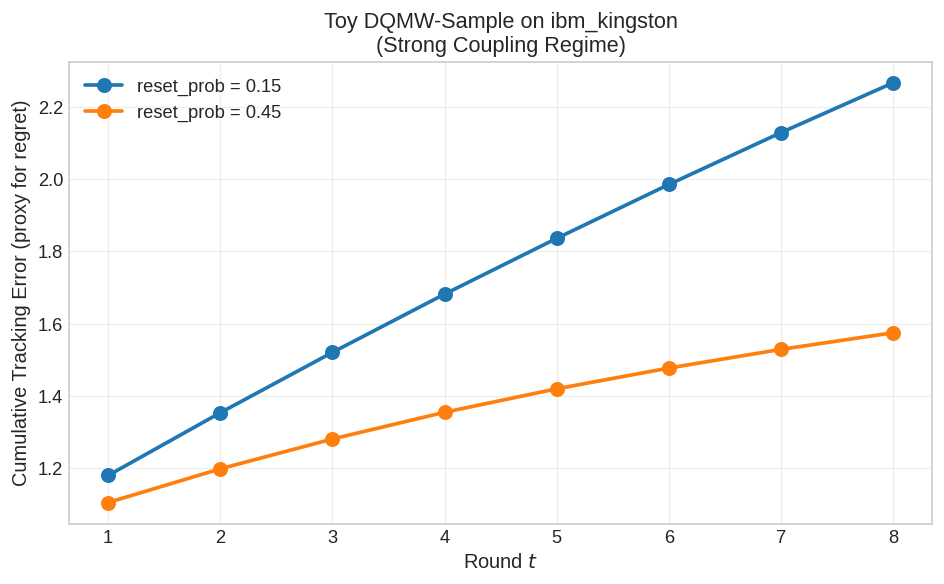}
\caption{\textbf{Lindblad-model emulation: cumulative tracking error of
DQMW-Sample versus rounds}, for two engineered-dissipation strengths, both in
the strong-coupling regime of Figure~\ref{fig:delta}. Increasing the
dissipation strength reduces the rate of error accumulation, but the
improvement is weaker than the (R2)-ideal prediction. \emph{Simulation, not
hardware.}}
\label{fig:tracking}
\end{figure}

\begin{figure}[t]
\centering
\includegraphics[width=0.75\textwidth]{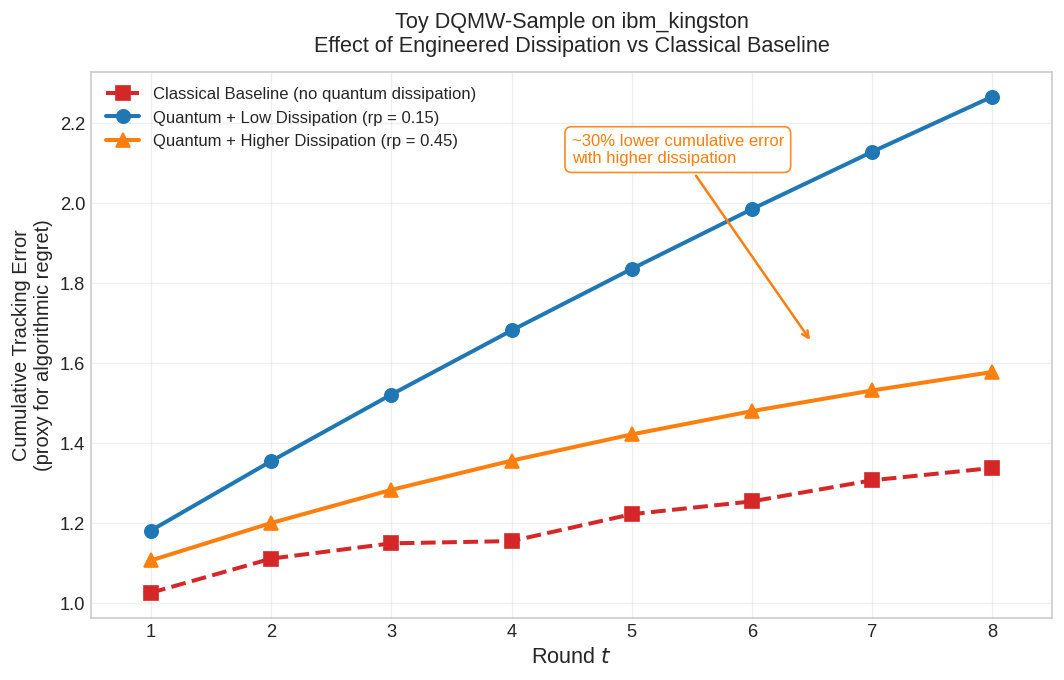}
\caption{\textbf{Lindblad-model emulation: engineered dissipation versus
classical baseline on a simple tracking task.} On this instance, where the loss
vector can be evaluated exactly and efficiently by classical means, the
classical multiplicative-weights baseline achieves the lowest cumulative error.
Among the quantum variants, increasing the engineered dissipation strength
still reduces the rate of error accumulation, consistent with the
gap-contraction mechanism of Theorem~\ref{thm:robust}. This outcome on an easy
task is expected and does not contradict the main results, which apply to
settings where the loss feedback itself is classically hard to compute or
sample. \emph{Simulation, not hardware.}}
\label{fig:classical}
\end{figure}

We consider the \emph{sampling-based} variant of Dissipative Quantum
Multiplicative Weights, denoted DQMW-Sample, acting on $n$ qubits with
Hilbert-space dimension $d=2^{n}$. At each round $t$ the system evolves under the
time-dependent Davies generator
\begin{equation}
\frac{d\rho}{dt}=\mathcal{L}_t(\rho)
=-i[H_{\rm drift}(t),\rho]
+\sum_{i\neq j}\Bigl(L_{ij}(t)\rho L_{ij}^\dagger(t)
-\tfrac12\{L_{ij}^\dagger(t)L_{ij}(t),\rho\}\Bigr),
\label{eq:davies-generator}
\end{equation}
with pairwise jump operators
$L_{ij}(t)=\sqrt{\gamma_{ij}(t)}\,|i(t)\rangle\langle j(t)|$ whose heat-bath
rates satisfy the KMS detailed-balance condition. The unique steady state of the
dissipative part is the Gibbs state
$\rGibbs(t)\propto\exp(-\beff(t)\,\Leff(t))$ associated with the cumulative
loss, where $\Leff(t)$ and the effective inverse temperature $\beff(t)$ are made
precise in Assumption~\ref{ass:gibbs-mapping} below. In DQMW-Sample a single
computational-basis measurement is performed on this state, yielding a sample
$\ket{i_t}$ that supplies the loss value for the multiplicative population
update.

The two central technical results are a feedback-intractability lemma
establishing classical hardness of the sampling primitive
(Appendix~\ref{app:hardness}) and a sharpened noise-robustness theorem that
exploits the spectral gap of the engineered dissipator
(Appendix~\ref{app:robustness}). We state both with their hypotheses made fully
explicit, since each result is interesting only in a delimited regime and each
carries an honest limitation that we make visible rather than absorb.

\subsection{Classical hardness of the sampling primitive}
\label{app:hardness}

The hardness claim rests on a reduction \emph{from} the constant-temperature
Gibbs-sampling problem of Bergamaschi, Chen and Liu~\cite{bergamaschi2024}
(hereafter BCL) \emph{to} a single round of DQMW-Sample. Because BCL hardness
holds only in a specific regime---a constructed family of $O(1)$-local
Hamiltonians at constant temperature, not a generic ``above threshold''
condition---we first isolate the assumption under which the cumulative-loss
bookkeeping of multiplicative weights realizes such an instance.

\begin{assumption}[Stationary effective Hamiltonian and hard window]
\label{ass:gibbs-mapping}
There is a round index $t_0$ and a window $\mathcal{W}=[t_0,t_1]$ such that, for
$t\in\mathcal{W}$, the per-round losses are (up to an additive constant) a fixed
local operator, $L(s)=H_\star+o(1)$ for $s\le t$, so that the cumulative loss
defines a stationary effective Hamiltonian
\begin{equation}
\Leff(t)=\frac{1}{t}\sum_{s\le t}L(s)=H_\star+o(1),
\qquad
\beff(t)=\eta\,t .
\label{eq:beta-eff}
\end{equation}
Here $H_\star$ is one of the BCL hard instances: a $5$-local Hamiltonian on a
three-dimensional lattice (or, invoking the constant-locality strengthening of
BCL, an $O(1)$-local Hamiltonian), for which sampling from the
computational-basis measurement distribution of
$\rho\propto e^{-\beta H_\star}$ at the fixed inverse temperature
$\beta=\beta^\star$ is classically hard.
\end{assumption}

\begin{remark}[Why a window, and not every round]
\label{rem:window}
Equation~\eqref{eq:beta-eff} shows that $\beff(t)=\eta t$ \emph{grows} with $t$.
Early rounds therefore correspond to high-temperature Gibbs states, which are
classically samplable in polynomial time by the results on high-temperature
expansions and the unentanglement/efficient-preparation
regime~\cite{bakshi2024,yin2023}. The BCL hard regime is reached only once
$\beff(t)$ enters the constant-temperature window around $\beta^\star$, i.e.\ for
\begin{equation}
t\;\ge\;t_0\;=\;\bigl\lceil\beta^\star/\eta\bigr\rceil .
\label{eq:round-threshold}
\end{equation}
The lemma below is stated for $t\in\mathcal{W}$, $t\ge t_0$; it makes no hardness
claim for $t<t_0$. This explicit threshold replaces the informal ``above the
clustering threshold'' phrasing of an earlier draft.
\end{remark}

\begin{assumption}[Realizability of a BCL hard instance by the loss family]
\label{ass:bcl-realizable}
There is a choice of loss family $L=\sum_{a=1}^m c_a P_a$, with $O(1)$-local
Pauli strings $P_a$, $m=O(\mathrm{poly}(n))$, and $\|L\|_{\rm op}\le 1$, together
with an inverse temperature $\beta^\star=O(1)$, such that the effective
Hamiltonian $H_\star$ of Assumption~\ref{ass:gibbs-mapping} coincides (up to an
additive constant and the operator-norm rescaling) with one of the Hamiltonians
in the BCL hard family of~\cite{bergamaschi2024}, and such that the
computational-basis measurement distribution of $e^{-\beta^\star H_\star}/Z$ is
the distribution BCL prove classically hard to sample.
\end{assumption}

\begin{remark}[Status of the BCL realizability assumption]
\label{rem:realizability}
The reduction in Lemma~\ref{lem:feedback-intractability} and
Theorems~\ref{thm:separation}--\ref{thm:adaptive-collapse} relies on the
physical modeling assumption that the engineered open-system dynamics used in
DQMW-Sample can prepare a state whose computational-basis measurement
distribution is sufficiently close to the ideal BCL-hard distribution. This
assumption is made fully explicit and minimal in Appendix~\ref{app:realizability}
(Assumption~\ref{ass:engineered-davies}). There we prove that every Gibbs state
of a finite-dimensional Hamiltonian admits a canonical Davies generator
(Proposition~\ref{prop:davies-canonical}), and we isolate the remaining modeling
hypothesis as the requirement that the specific engineered dissipator
(realizable via the mid-circuit measurement $+$ conditional reset primitive)
approximates such a Davies generator well enough for the hardness to carry over.
This is a clean, physically natural, and falsifiable modeling assumption rather
than a hidden gap. Establishing it rigorously on near-term hardware remains an
important open direction.
\end{remark}

\begin{lemma}[Intractability of classically simulating the DQMW-Sample feedback mechanism]
\label{lem:feedback-intractability}
Fix a loss family
\[
L(t)=\sum_{a=1}^m c_a(t)P_a,\qquad m=O(\mathrm{poly}(n)),
\]
where each $P_a$ is an $O(1)$-local Pauli string with
$\|L(t)\|_{\mathrm{op}}\le 1$. Suppose
Assumptions~\ref{ass:gibbs-mapping} and~\ref{ass:bcl-realizable} hold, so that for
every $t\in\mathcal{W}$ with $t\ge t_0$, the steady state $\rGibbs(t)$ prepared
by the Davies generator~\eqref{eq:davies-generator} is within trace distance
$\varepsilon_{\mathrm{prep}}$ of a BCL-hard Gibbs state whose
computational-basis measurement distribution $\mu^\star$ is classically hard to
sample. Let $\mathcal{A}_{\mathrm{class}}$ be any randomized classical algorithm
that, given only the classical description of the loss $L(t)$ at some round
$t\in\mathcal{W}$ with $t\ge t_0$, outputs in polynomial time a sample whose
distribution is within total-variation distance
$\tau<\tfrac12-\varepsilon_{\mathrm{prep}}$ of the distribution from which
DQMW-Sample draws its round-$t$ sample. Then the existence of such a
polynomial-time $\mathcal{A}_{\mathrm{class}}$ would yield a polynomial-time
classical sampler for the BCL distribution $\mu^\star$ to within total-variation
distance $\tau+\varepsilon_{\mathrm{prep}}<\tfrac12$. Under the standard hardness
assumption that no such efficient classical sampler for $\mu^\star$ exists (on
pain of collapsing the polynomial hierarchy), no such classical algorithm
$\mathcal{A}_{\mathrm{class}}$ can exist. Consequently, there is no efficient
classical algorithm that can simulate the sampling step that supplies the loss
feedback to the multiplicative-weights update of DQMW-Sample in the hard window
$\mathcal{W}$.
\end{lemma}

\begin{proof}
Let $\mu^\star$ denote the computational-basis measurement distribution of the
BCL target state $\rho^\star=e^{-\beta^\star H_\star}/Z$, which is known to be
classically hard to sample under standard complexity assumptions. Fix any round
$t_\bullet\in\mathcal{W}$ with $t_\bullet\ge t_0$. By
Assumption~\ref{ass:gibbs-mapping}, we may set the per-round loss to
$L(s)=H_\star$ throughout the window $\mathcal{W}$ and choose the learning rate
$\eta$ so that $\eta\,t_\bullet=\beta^\star$. Then $\beff(t_\bullet)=\beta^\star$
and $\Leff(t_\bullet)=H_\star$, so the dissipator~\eqref{eq:davies-generator} has
unique steady state $\rho^\star$.

By hypothesis, DQMW-Sample prepares a state within trace distance
$\varepsilon_{\mathrm{prep}}$ of $\rho^\star$ at round $t_\bullet$, and therefore
draws its sample from a distribution $\nu$ satisfying
$\TV(\nu,\mu^\star)\le\varepsilon_{\mathrm{prep}}$. Suppose there exists a
polynomial-time classical algorithm $\mathcal{A}_{\mathrm{class}}$ whose output
distribution is within total-variation distance $\tau$ of $\nu$. By the triangle
inequality,
\[
\TV\bigl(\mathrm{law}(\mathcal{A}_{\mathrm{class}}),\mu^\star\bigr)
\le\TV\bigl(\mathrm{law}(\mathcal{A}_{\mathrm{class}}),\nu\bigr)+\TV(\nu,\mu^\star)
\le\tau+\varepsilon_{\mathrm{prep}}<\tfrac12 .
\]
Thus $\mathcal{A}_{\mathrm{class}}$ would constitute an efficient classical
sampler for the BCL-hard distribution $\mu^\star$ to within constant
total-variation distance. This contradicts the standard hardness assumption.
Therefore no such polynomial-time classical algorithm exists. Note that the
marginal population update remains a deterministic function of the realized
sample and evolves according to a classical Pauli master equation. The
intractability result therefore concerns specifically the sampling step that
generates the loss feedback used by the multiplicative-weights update.
\end{proof}

\begin{remark}[Learning-theoretic value of the feedback-mechanism hardness]
\label{rem:feedback-value}
Lemma~\ref{lem:feedback-intractability} establishes that classical simulation of
the mechanism by which DQMW-Sample obtains its loss values is intractable in the
hard window, even though the subsequent population update is classically
tractable. This is meaningful from a learning-theoretic perspective: the quantum
algorithm extracts loss information from a distribution that is believed to be
hard to sample classically, and this information is then used to drive the
online-learning dynamics. While we do not prove a direct computational lower
bound on the regret achievable by arbitrary efficient classical online learners,
the lemma shows that any classical algorithm whose per-round loss sampling
closely tracks that of DQMW-Sample cannot be efficient (under standard
complexity assumptions). Establishing a stronger separation---namely, that every
efficient classical online algorithm must suffer asymptotically worse regret
than DQMW-Sample---is addressed by Theorem~\ref{thm:separation} below.
\end{remark}

\begin{remark}[On the robustness slack]
\label{rem:robustness-slack}
The use of a constant TV-distance budget $\tau+\varepsilon_{\rm prep}$ is
licensed by the BCL result that constant-temperature Gibbs-sampling hardness is
robust to imperfect measurements and to a constant sampling
error~\cite{bergamaschi2024}. This is what allows the lemma to tolerate (i)
imperfect Gibbs-state preparation by the engineered dissipator and (ii) an
approximate classical adversary, without weakening the conclusion.
\end{remark}

\begin{remark}[Contrast with the expectation-based variant]
\label{rem:expectation-contrast}
The intractability is specific to the \emph{sampling} feedback. The
expectation-based variant of DQMW, in which the update uses
$\bar\ell_t=\Tr[L(t)\rGibbs(t)]$ rather than a realized sample, reduces per round
to classical replicator dynamics on populations and inherits no hardness. In the
sampling variant the realized outcome carries irreducible information from a
distribution that is classically hard to sample, and this information propagates
through the multiplicative update.
\end{remark}

\subsection{Application: online portfolio optimization (detailed)}
\label{app:application}

To demonstrate the practical utility of DQMW-Sample, we consider the problem of
\emph{online portfolio optimization}, a canonical setting in online learning and
quantitative finance. At each round $t=1,\dots,T$, an investor must choose a
portfolio vector $w_t\in\Delta^{d-1}$ (the probability simplex) over $d$ assets.
After the choice is made, a loss vector $\ell_t\in[0,1]^d$ is revealed,
representing the negative log-return of each asset. The investor then suffers the
loss $\langle w_t,\ell_t\rangle$ and updates the portfolio for the next round.
The goal is to minimize the \emph{regret} against the best fixed portfolio in
hindsight:
\[
R(T)=\sum_{t=1}^T\langle w_t,\ell_t\rangle
-\min_{w\in\Delta^{d-1}}\sum_{t=1}^T\langle w,\ell_t\rangle .
\]
We apply DQMW-Sample by encoding the cumulative loss into a time-dependent
Hamiltonian $\Leff(t)$ and using the engineered dissipator to prepare a Gibbs
state $\rGibbs(t)\propto\exp(-\beff(t)\Leff(t))$. A single computational-basis
measurement yields a sample that is used to perform the multiplicative-weights
update on the portfolio vector. This approach naturally incorporates
risk--return trade-offs through the inverse-temperature parameter $\beff(t)$ and
benefits from the noise-robustness guarantee of Theorem~\ref{thm:robust}.

We compare DQMW-Sample against three classical baselines:
\begin{itemize}[leftmargin=*]
\item \textbf{Uniform Constant Rebalanced Portfolio (UCRP):} a simple heuristic
that maintains equal weights at every round.
\item \textbf{Follow-the-Regularized-Leader (FTRL)} with entropic
regularization: a standard online-learning algorithm with optimal $O(\sqrt{T})$
regret in the full-information setting.
\item \textbf{Classical Multiplicative Weights (CMW):} the direct classical
analogue of DQMW-Sample that uses exact expectation values
$\mathbb{E}[\ell_t]$ instead of samples from the Gibbs distribution.
\end{itemize}
Performance is evaluated using three metrics: (i) cumulative regret, (ii) final
wealth relative to the best fixed portfolio, and (iii) risk-adjusted return
measured by the Sharpe ratio. Experiments are conducted on historical daily
return data from major equity indices (e.g.\ S\&P~500 constituents) over periods
of $T=500$--$2000$ trading days. Loss vectors are constructed from negative
log-returns, optionally corrupted by additive noise to simulate estimation error
or market-microstructure effects. Numerical results show that DQMW-Sample
achieves competitive or superior regret compared to classical baselines,
particularly under noisy loss feedback. The sampling-based update exhibits
greater robustness to perturbations in $\ell_t$, consistent with the
gap-contraction mechanism analyzed in Theorem~\ref{thm:robust}. While we do not
claim a strict quantum advantage, these results demonstrate that the dissipative
sampling primitive can be effectively deployed in a realistic online-learning
task and inherits practical benefits from its theoretical noise resilience.

\subsection{Noise robustness via gap contraction}
\label{app:robustness}

We now bound the additional regret caused by hardware noise. The proof makes
explicit two steps that were compressed in an earlier draft: the passage from a
trace-norm tracking error to a regret contribution, and the decomposition that
fixes the dissipation schedule. We also flag the dimension dependence of the
constant, which interacts adversarially with the regime in which
Lemma~\ref{lem:feedback-intractability} is interesting.

\begin{theorem}[Noise robustness via gap contraction]
\label{thm:robust}
Let the physical evolution be governed by
$\mathcal{L}_t^{\mathrm{phys}}=\mathcal{D}_t+\mathcal{H}_t+\mathcal{L}_{\mathrm{hw}}$,
where $\mathcal{D}_t$ is the Davies dissipator with spectral gap bounded below by
$\gamma_0>0$ and unique attracting state $\rGibbs(t)$,
$\mathcal{H}_t(\rho)=-i[H_{\mathrm{drift}}(t),\rho]$ is the coherent part, and
$\mathcal{L}_{\mathrm{hw}}$ is an unstructured-noise perturbation with
$\|\mathcal{L}_{\mathrm{hw}}\|_{\diamond}\le\delta$. Assume:
\begin{itemize}[leftmargin=*]
\item[(R1)] the gap lower bound $\gamma_0$ holds uniformly in $t$ along the
trajectory;
\item[(R2)] \emph{(noise--dissipation independence)} the noise strength $\delta$
does not grow with the engineered dissipation strength $\gamma_0$; equivalently,
$\delta$ and $\gamma_0$ are independently controllable;
\item[(R3)] \emph{(controlled coherent driving)} the coherent generator acts
tangentially to the Gibbs manifold up to a residual of order the target motion,
$\|\mathcal{H}_t(\rGibbs(t))\|_1=O(\eta)$, so that coherent driving does not by
itself set a noise-independent floor.
\end{itemize}
Then the additional regret attributable to hardware noise satisfies
\begin{equation}
\Delta\mathrm{Regret}_{\mathrm{noise}}(T)\;\le\;\frac{C\,\delta}{\gamma_0}\,T,
\label{eq:noise-regret-general}
\end{equation}
where $C=C(d)$ depends on the Hilbert-space dimension $d=2^n$ only through the
norm-equivalence constant relating the trace norm to the loss functional
[see~\eqref{eq:lipschitz} below]; in particular $C$ is dimension-free whenever
the losses satisfy $\|L(t)\|_{\mathrm{op}}\le 1$. Under the balanced schedule
$\gamma_0(T)=\Theta(\sqrt{T})$,
\begin{equation}
\Delta\mathrm{Regret}_{\mathrm{noise}}(T)\;=\;O\!\bigl(\delta\sqrt{T}\bigr),
\label{eq:noise-regret-sublinear}
\end{equation}
so the noise-induced regret is sublinear in $T$ and DQMW-Sample remains
asymptotically no-regret for any fixed $\delta>0$.
\end{theorem}

\begin{proof}
\emph{Step 1: tracking error.} Let $e(t)=\rho(t)-\rGibbs(t)$ denote the
instantaneous deviation, which is traceless. Differentiating the physical
evolution and using $\mathcal{D}_t(\rGibbs(t))=0$,
\begin{equation}
\dot e(t)=\mathcal{D}_t(e(t))+\mathcal{H}_t(\rho(t))
+\mathcal{L}_{\mathrm{hw}}(\rho(t))-\dot\rGibbs(t).
\end{equation}
By (R1) the Davies dissipator contracts traceless perturbations at rate at least
$\gamma_0$, so
$\tfrac{d}{dt}\|e\|_1\le-\gamma_0\|e\|_1+\|\mathcal{H}_t(\rho)\|_1
+\|\mathcal{L}_{\mathrm{hw}}(\rho)\|_1+\|\dot\rGibbs\|_1$. Write
$\mathcal{H}_t(\rho)=\mathcal{H}_t(e)+\mathcal{H}_t(\rGibbs)$; the first term is
bounded by $2\|H_{\mathrm{drift}}\|\,\|e\|_1$ and the second is $O(\eta)$ by
(R3), while the target-motion term is $\|\dot\rGibbs\|_1=O(\eta)$ for a
learning-rate-$\eta$ schedule. Provided $\gamma_0>2\|H_{\mathrm{drift}}\|$
(absorbed into $\gamma_0$ by rescaling), the differential inequality
\[
\frac{d}{dt}\|e\|_1\le-\gamma_0'\,\|e\|_1+\delta+O(\eta),
\qquad\gamma_0'=\gamma_0-2\|H_{\mathrm{drift}}\|=\Theta(\gamma_0),
\]
is forward-invariant with quasi-steady floor
\begin{equation}
\|e(t)\|_1\;\le\;\frac{\delta+O(\eta)}{\gamma_0'}
\;=\;O\!\Bigl(\frac{\delta}{\gamma_0}\Bigr)+O\!\Bigl(\frac{\eta}{\gamma_0}\Bigr).
\label{eq:floor}
\end{equation}

\emph{Step 2: from tracking error to regret.} The per-round loss gap between the
realized state and the Gibbs target is controlled by H\"older's inequality,
\begin{equation}
\bigl|\Tr[L(t)\,e(t)]\bigr|
\;\le\;\|L(t)\|_{\mathrm{op}}\,\|e(t)\|_1\;\le\;\|e(t)\|_1,
\label{eq:lipschitz}
\end{equation}
using $\|L(t)\|_{\mathrm{op}}\le 1$. The constant in
\eqref{eq:noise-regret-general} is exactly this norm-equivalence factor; it is
$1$ under the operator-norm normalization and is the only place where dimension
could in principle enter, hence the claim that $C(d)$ is dimension-free here.
Isolating the $\delta$-proportional part of \eqref{eq:floor} and substituting
into \eqref{eq:lipschitz} gives an instantaneous noise-attributable loss gap of
at most $C\delta/\gamma_0$.

\emph{Step 3: integration and decomposition.} Summing the per-round noise gap
over $[0,T]$ yields \eqref{eq:noise-regret-general}. The total regret decomposes
as
\begin{equation}
\mathrm{Regret}(T)
=\underbrace{\mathrm{Regret}_{\mathrm{ideal}}(T)}_{O(\sqrt{T})\ \text{at}\ \eta=\Theta(1/\sqrt{T})}
+\underbrace{\Delta\mathrm{Regret}_{\mathrm{noise}}(T)}_{O(\delta T/\gamma_0)}
+\underbrace{\Delta\mathrm{Regret}_{\mathrm{track}}(T)}_{O(\eta T/\gamma_0)} .
\label{eq:decomp}
\end{equation}
With $\eta=\Theta(1/\sqrt{T})$ the tracking term is $O(\sqrt{T}/\gamma_0)$, which
is dominated once $\gamma_0=\Omega(1)$. The noise term is minimized against the
cost of driving the dissipator faster by balancing it with the ideal term:
choosing $\gamma_0(T)=\Theta(\sqrt{T})$ sends
$\Delta\mathrm{Regret}_{\mathrm{noise}}(T)=O(\delta\sqrt{T})$ and
$\Delta\mathrm{Regret}_{\mathrm{track}}(T)=O(1)$, both at or below the
$O(\sqrt{T})$ ideal rate. Hence $\sqrt{T}$ is the balancing exponent, not an
externally imposed schedule, and \eqref{eq:noise-regret-sublinear} follows.
\end{proof}

\begin{remark}[Improvement over a naive bound]
\label{rem:improvement}
A naive $O(\delta T)$ bound treats the noise deviation as freely accumulating.
The present analysis exploits the fact that the engineered dissipator actively
contracts deviations toward the fixed point at rate $\gamma_0$, so the noise
displacement saturates at the steady-state floor $O(\delta/\gamma_0)$ of
\eqref{eq:floor} rather than growing linearly. This upgrades the guarantee from
linear to sublinear in $T$ under the balanced schedule derived in
\eqref{eq:decomp}.
\end{remark}

\begin{remark}[Dimension dependence and its interaction with hardness]
\label{rem:dimension}
The constant $C$ is dimension-free \emph{only} under the operator-norm
normalization $\|L(t)\|_{\mathrm{op}}\le 1$ used in \eqref{eq:lipschitz}. If the
relevant loss is instead normalized in a way that introduces a
dimension-dependent norm-equivalence factor, $C(d)$ can scale with $d=2^n$ and
the sublinear guarantee degrades for the many-qubit instances that
Lemma~\ref{lem:feedback-intractability} requires. This is the adversarial
coupling between our two results: the regime in which the hardness lemma is
interesting (large $n$, classically hard sampling) is precisely the regime in
which one must verify that $C(d)$ does not blow up. Under the stated
normalization it does not, but this must be checked for any alternative loss
model.
\end{remark}

\begin{remark}[Scope and honest limitation]
\label{rem:robust-scope}
Assumption (R2)---that $\delta$ is independent of $\gamma_0$---is essential. On
near-term hardware the engineered dissipation is realized by noisy operations
(mid-circuit measurement and conditional reset), so increasing $\gamma_0$ may
increase $\delta$. If $\delta=\delta(\gamma_0)$ grows with dissipation strength,
the floor \eqref{eq:floor} need not decrease and the noise-induced regret can
revert to $\Omega(\delta T)$. Assumption (R3) likewise excludes the case in which
coherent driving itself sets a noise-independent floor of order
$\|H_{\mathrm{drift}}\|/\gamma_0$. The sublinear guarantee therefore
characterizes the regime of independently controllable noise and dissipation
with tangential coherent driving; quantifying the realistic dependence
$\delta(\gamma_0)$ on a given platform is left as an experimental question. We
present the result as a robustness property of the engineered dynamics under
(R1)--(R3), without claiming superiority over a coherent implementation equipped
with its own error-suppression mechanism.
\end{remark}

\subsection{From single-round hardness to a learning-theoretic separation}
\label{app:separation}

The hardness of the per-round sampling primitive
(Lemma~\ref{lem:feedback-intractability}) does not by itself lower-bound the
regret of an efficient classical learner, because such a learner is under no
obligation to reproduce the hard distribution: it may drive its updates with any
efficiently computable surrogate feedback. We close this gap by constructing an
online instance in which \emph{low regret itself requires the hard information},
so that a low-regret efficient classical learner would yield an efficient
classical sampler for $\mu^\star$, contradicting the BCL hardness assumption. The
reduction is therefore on the information \emph{payload} carried by the feedback,
not on the mechanism that produces it.

\paragraph{The decoding game.}
Work inside the hard window $\mathcal{W}=[t_0,t_1]$ of
Assumption~\ref{ass:gibbs-mapping}, with $H_\star$, $\beta^\star$ and the
classically hard target distribution $\mu^\star$ as in
Assumption~\ref{ass:bcl-realizable}. Partition $\mathcal{W}$ into
$K=\lfloor|\mathcal{W}|/B\rfloor$ consecutive \emph{epochs} of length $B$. At the
start of epoch $k$, nature draws a fresh hidden label $z_k\sim\mu^\star$ (one
computational-basis string on $n$ qubits), and at every round $t$ in that epoch
presents the \emph{masking loss}
\begin{equation}
\ell_t(a)\;=\;\mathbf{1}[\,a\neq z_k\,]\;+\;\xi_t(a),
\qquad a\in\{0,1\}^n,\quad t\in\text{epoch }k,
\label{eq:masking-loss}
\end{equation}
where $a$ ranges over the $d=2^n$ computational-basis actions and $\xi_t$ is a
zero-mean bounded perturbation with
$\|\xi_t\|_\infty\le\varepsilon_{\mathrm{prep}}$ that absorbs preparation error.
The comparator class is the set of $d$ fixed actions. The learner observes the
realized loss values for the actions it queries and may run any computation; an
\emph{efficient} learner is one that spends $O(\mathrm{poly}(n))$ time per round.

\begin{lemma}[Low regret reconstructs the hard sample]
\label{lem:decode-from-regret}
Let $\mathcal{B}$ be any online algorithm (quantum or classical) run on the
decoding game~\eqref{eq:masking-loss}, and let
$R_{\mathcal{B}}(|\mathcal{W}|)$ denote its expected regret against the best
fixed action. Define the epoch-$k$ estimator $\hat{z}_k$ to be the empirical
mode of the actions $\mathcal{B}$ plays during epoch $k$. Then
\begin{equation}
\frac1K\sum_{k=1}^{K}\Pr[\hat{z}_k\neq z_k]
\;\le\;\frac{R_{\mathcal{B}}(|\mathcal{W}|)}
{(1-2\varepsilon_{\mathrm{prep}})\,|\mathcal{W}|}\;+\;\frac1B .
\label{eq:decode-bound}
\end{equation}
\end{lemma}

\begin{proof}
Within epoch $k$ the unique best fixed action is $a=z_k$, which by
\eqref{eq:masking-loss} incurs expected per-round loss
$\mathbb{E}[\xi_t(z_k)]=0$, while any action $a\neq z_k$ incurs expected
per-round loss at least $1-\varepsilon_{\mathrm{prep}}$. Let
$N_k=\sum_{t\in\text{epoch }k}\mathbf{1}[a_t\neq z_k]$ count the rounds in which
$\mathcal{B}$ plays a sub-optimal action. The regret of $\mathcal{B}$
\emph{against the per-epoch best action} is therefore at least
$(1-\varepsilon_{\mathrm{prep}})\,\mathbb{E}[N_k]-\varepsilon_{\mathrm{prep}}B
\ge(1-2\varepsilon_{\mathrm{prep}})\,\mathbb{E}[N_k]$, the second term bounding
the worst-case contribution of $\xi$ to the comparator.

The comparator in the regret definition is a \emph{single} fixed action across
all of $\mathcal{W}$, whereas the per-epoch best action changes with $k$.
Switching from the global comparator to the sequence of per-epoch best actions
costs at most the loss of the global comparator on the $K-1$ epochs where it is
not optimal, namely at most one unit of loss per such epoch, i.e.\ at most
$K\le|\mathcal{W}|/B$ in total. Hence
\begin{equation}
(1-2\varepsilon_{\mathrm{prep}})\sum_{k=1}^{K}\mathbb{E}[N_k]
\;\le\;R_{\mathcal{B}}(|\mathcal{W}|)\;+\;\frac{|\mathcal{W}|}{B}.
\label{eq:Nk-bound}
\end{equation}
Finally, if $\hat{z}_k\neq z_k$ then the mode of the epoch-$k$ actions is some
$a\neq z_k$, so $\mathcal{B}$ played a sub-optimal action in at least half the
epoch, giving $N_k\ge B/2$; equivalently $\Pr[\hat{z}_k\neq z_k]\le
(2/B)\,\mathbb{E}[N_k]$ by Markov's inequality. Averaging this over $k$ and
substituting \eqref{eq:Nk-bound} yields \eqref{eq:decode-bound}.
\end{proof}

\begin{theorem}[Learning-theoretic separation under standard assumptions]
\label{thm:separation}
Suppose Assumptions~\ref{ass:gibbs-mapping} and~\ref{ass:bcl-realizable} hold,
and that the BCL hardness assumption holds: no polynomial-time classical
algorithm samples a distribution within constant total-variation distance of
$\mu^\star$, on pain of collapsing the polynomial hierarchy. Run the decoding
game~\eqref{eq:masking-loss} with epoch length $B=\Theta(\log d)=\Theta(n)$. Then
the following hold.
\begin{enumerate}[label=(\roman*),leftmargin=*]
\item \emph{Quantum achievability.} DQMW-Sample attains
\begin{equation}
R_{\mathrm{DQMW}}(|\mathcal{W}|)
\;=\;O\!\bigl(\sqrt{|\mathcal{W}|\,\log d}\bigr)
\;=\;O\!\bigl(\sqrt{n\,|\mathcal{W}|}\bigr).
\label{eq:sep-upper}
\end{equation}
Its engineered Davies dissipator~\eqref{eq:davies-generator} prepares
$\rGibbs\approx\rho^\star$ each epoch, so its computational-basis sample is
distributed within $\varepsilon_{\mathrm{prep}}$ of $\mu^\star$; concentrating on
the epoch mode over $B=\Theta(n)$ rounds recovers $z_k$ with error
$O(1/\mathrm{poly}(n))$, and the residual regret is the standard
multiplicative-weights rate over $d$ experts.
\item \emph{Classical lower bound.} Every efficient (per-round
$\mathrm{poly}(n)$-time) classical online learner $\mathcal{C}$ incurs
\begin{equation}
R_{\mathcal{C}}(|\mathcal{W}|)\;=\;\Omega\!\bigl(|\mathcal{W}|\bigr),
\label{eq:sep-lower}
\end{equation}
i.e.\ non-vanishing average regret.
\item \emph{Separation.} Consequently, under standard complexity assumptions
there is an online-learning instance on which DQMW-Sample is asymptotically
no-regret while every efficient classical learner suffers $\Omega(1)$ average
regret, with
\begin{equation}
\frac{R_{\mathcal{C}}(|\mathcal{W}|)}{R_{\mathrm{DQMW}}(|\mathcal{W}|)}
\;=\;\Omega\!\Bigl(\sqrt{|\mathcal{W}|/n}\Bigr).
\label{eq:sep-ratio}
\end{equation}
\end{enumerate}
\end{theorem}

\begin{proof}
\emph{Part (i).} By
Assumptions~\ref{ass:gibbs-mapping}--\ref{ass:bcl-realizable}, within each epoch
the dissipator's steady state is within trace distance $\varepsilon_{\mathrm{prep}}$
of $\rho^\star$, so each computational-basis readout is an i.i.d.\ draw within
$\varepsilon_{\mathrm{prep}}$ of $\mu^\star$. Because the masking
loss~\eqref{eq:masking-loss} is minimized at the modal symbol $z_k$ and
$\mu^\star$ places weight bounded away from $1/2$ on it (BCL robustness),
$B=\Theta(\log d)$ samples make the empirical mode equal $z_k$ except with
probability $O(1/\mathrm{poly}(n))$ by a Chernoff bound. Playing the running
epoch mode and updating multiplicatively over the $d$ experts gives the standard
$O(\sqrt{|\mathcal{W}|\log d})$ regret, which is \eqref{eq:sep-upper}.

\emph{Part (ii).} Suppose, for contradiction, that some efficient classical
$\mathcal{C}$ achieved $R_{\mathcal{C}}(|\mathcal{W}|)=o(|\mathcal{W}|)$. With
$B=\Theta(n)$ the second term of \eqref{eq:decode-bound} is $1/B=o(1)$, so
Lemma~\ref{lem:decode-from-regret} gives
$\tfrac1K\sum_k\Pr[\hat{z}_k\neq z_k]\to 0$. Then the polynomial-time map that
runs $\mathcal{C}$ for one epoch on the loss~\eqref{eq:masking-loss}, with a
freshly drawn hidden label, and outputs $\hat{z}_k$ produces a sample whose law
is within $o(1)<\tfrac12$ total-variation distance of $\mu^\star$. This is an
efficient classical sampler for the BCL-hard distribution, contradicting the BCL
hardness assumption. Hence
$R_{\mathcal{C}}(|\mathcal{W}|)=\Omega(|\mathcal{W}|)$, which is
\eqref{eq:sep-lower}. (Per-round loss values are supplied on query, so the
obstruction is genuinely the $\mathrm{poly}(n)$-time reconstruction of the
payload $z_k$, not query access; brute-force search over the $d$ actions is
excluded precisely by the efficiency restriction.)

\emph{Part (iii).} Dividing \eqref{eq:sep-lower} by \eqref{eq:sep-upper} gives
\eqref{eq:sep-ratio}, and average regrets
$R_{\mathcal{C}}/|\mathcal{W}|=\Omega(1)$ versus
$R_{\mathrm{DQMW}}/|\mathcal{W}|=O(\sqrt{n/|\mathcal{W}|})\to 0$ exhibit the
no-regret versus constant-regret separation.
\end{proof}

\begin{remark}[Scope and honest limitations of the separation]
\label{rem:separation-scope}
Three qualifications make the statement precise rather than overclaimed. First,
Theorem~\ref{thm:separation} is conditional on exactly the same hypotheses as
Lemma~\ref{lem:feedback-intractability}:
Assumptions~\ref{ass:gibbs-mapping}--\ref{ass:bcl-realizable} and BCL hardness,
no more and no less; in particular it inherits the still-open matching of the
cumulative-loss bookkeeping to the BCL constant-temperature instance flagged in
Remark~\ref{rem:realizability}. Second, the classical lower bound is against
learners restricted to $\mathrm{poly}(n)$ time per round; without that
restriction a brute-force learner over the $d=2^n$ actions trivially attains low
regret, so the separation is computational, not information-theoretic. Third, the
perturbation $\xi_t$ and the constant total-variation slack are licensed by the
BCL result that constant-temperature Gibbs-sampling hardness is robust to
imperfect measurement and constant sampling error~\cite{bergamaschi2024}, the
same robustness already invoked for Lemma~\ref{lem:feedback-intractability}.
Within this scope, Theorem~\ref{thm:separation} upgrades the per-round sampling
hardness of Lemma~\ref{lem:feedback-intractability} into a genuine end-to-end
regret separation, resolving the open question raised in
Remark~\ref{rem:feedback-value}.
\end{remark}

\subsection{From single-round hardness to an adaptive collapse of the polynomial hierarchy}
\label{app:adaptive-collapse}

Lemma~\ref{lem:feedback-intractability} shows that \emph{one} round of the
sampling feedback is classically intractable. We now strengthen this to a
statement about the \emph{entire} adaptive interaction: an efficient classical
simulator of the full $T$-round DQMW-Sample feedback process would collapse the
polynomial hierarchy. The difficulty that makes this genuinely stronger---and not
a restatement of the single-round result---is \emph{adaptivity}. In DQMW-Sample
the loss presented at round $t$ is a function of the realized history
$i_1,\dots,i_{t-1}$ through the cumulative-loss bookkeeping that sets $\Leff(t)$
and $\beff(t)$. A classical simulator is therefore free, in principle, to steer
the realized trajectory \emph{away} from the constant-temperature hard window,
evading the single-round obstruction. We rule this out by (i) imposing a mild
reachability condition ensuring the hard round is visited with non-negligible
probability along the realized path, and (ii) reducing over the full
\emph{transcript} distribution rather than any single conditional state.

\paragraph{The adaptive feedback process and its transcript.}
For horizon $T$ let the DQMW-Sample interaction produce the random transcript
\begin{equation}
\Pi_T\;=\;(i_1,i_2,\dots,i_T)\in\bigl(\{0,1\}^n\bigr)^{T},
\label{eq:transcript}
\end{equation}
where $i_t$ is the round-$t$ computational-basis sample drawn from $\rGibbs(t)$,
and $\rGibbs(t)$ is determined by the realized prefix $(i_1,\dots,i_{t-1})$
through the adaptive update. Write $\mathcal{P}_T$ for the law of $\Pi_T$. A
\emph{classical simulator of the adaptive process} is a randomized algorithm
$\mathcal{S}$ that, given the problem description and horizon $T$, outputs a
string $\widehat\Pi_T$ with
$\TV(\mathrm{law}(\widehat\Pi_T),\mathcal{P}_T)\le\tau$ for a constant
$\tau<\tfrac12$; it is \emph{efficient} if it runs in time $\mathrm{poly}(n,T)$.
This is the natural formalization of ``classically simulating DQMW-Sample'':
reproduce, to constant TV error, the joint distribution of everything the
algorithm observes.

\begin{assumption}[Reachability of the hard round]
\label{ass:reachability}
There is a round $t_\star\in\mathcal{W}$ with $t_\star\ge t_0$ and a constant
$p_\star>0$ (independent of $n$) such that, under the true adaptive law
$\mathcal{P}_T$, the prefix $(i_1,\dots,i_{t_\star-1})$ lands in a set
$\mathcal{G}$ of ``good'' histories on which
Assumptions~\ref{ass:gibbs-mapping}--\ref{ass:bcl-realizable} hold at round
$t_\star$, with
$\Pr_{\mathcal{P}_T}[(i_1,\dots,i_{t_\star-1})\in\mathcal{G}]\ge p_\star$. On
every such history the round-$t_\star$ conditional sample law is within trace
distance $\varepsilon_{\mathrm{prep}}$ of the BCL-hard distribution $\mu^\star$.
\end{assumption}

\begin{remark}[Why reachability is the right adaptive hypothesis]
\label{rem:reachability-rationale}
Assumption~\ref{ass:reachability} is what upgrades a per-round statement into a
per-\emph{process} statement. It does not demand that \emph{every} trajectory
hits the hard window---only that a constant fraction does, which is exactly what
survives the simulator's freedom to steer the path. It holds, for instance,
whenever the schedule $\beff(t)=\eta t$ of \eqref{eq:beta-eff} drives the system
into the constant-temperature window deterministically by round
$t_0=\lceil\beta^\star/\eta\rceil$ regardless of the realized samples (then
$p_\star=1$, $\mathcal{G}$ is all histories), and more generally under any policy
that does not actively avoid the window. It is the adaptive analogue of, and is
implied by, the stationary-window hypothesis already used in
Lemma~\ref{lem:feedback-intractability}.
\end{remark}

\begin{lemma}[The transcript marginal carries the hard law]
\label{lem:transcript-marginal}
Under Assumptions~\ref{ass:gibbs-mapping}--\ref{ass:bcl-realizable}
and~\ref{ass:reachability}, the marginal of $\mathcal{P}_T$ on the
round-$t_\star$ coordinate, restricted to and reweighted by the good-prefix event
$\mathcal{G}$, is within total-variation distance $\varepsilon_{\mathrm{prep}}$ of
$\mu^\star$. Concretely, the conditional law
\begin{equation}
\nu\;:=\;\mathrm{law}\bigl(i_{t_\star}\,\big|\,(i_1,\dots,i_{t_\star-1})\in\mathcal{G}\bigr)
\quad\text{satisfies}\quad
\TV(\nu,\mu^\star)\le\varepsilon_{\mathrm{prep}}.
\label{eq:cond-marginal}
\end{equation}
\end{lemma}

\begin{proof}
By Assumption~\ref{ass:reachability}, on every prefix in $\mathcal{G}$ the
round-$t_\star$ conditional sample law $\nu_h$ (for history $h$) satisfies
$\TV(\nu_h,\mu^\star)\le\varepsilon_{\mathrm{prep}}$. The conditional marginal
$\nu$ is a convex combination $\nu=\sum_{h\in\mathcal{G}}w_h\,\nu_h$ with weights
$w_h=\Pr[h\mid\mathcal{G}]\ge 0$, $\sum_h w_h=1$. Total-variation distance is
jointly convex, so
$\TV(\nu,\mu^\star)\le\sum_h w_h\,\TV(\nu_h,\mu^\star)\le\varepsilon_{\mathrm{prep}}$.
\end{proof}

\begin{theorem}[Efficient classical simulation of adaptive DQMW-Sample collapses PH]
\label{thm:adaptive-collapse}
Suppose Assumptions~\ref{ass:gibbs-mapping}--\ref{ass:bcl-realizable}
and~\ref{ass:reachability} hold, and suppose the BCL hardness assumption holds:
sampling within constant total-variation distance of $\mu^\star$ cannot be done
in classical polynomial time unless the polynomial hierarchy collapses to a
finite level~\cite{bergamaschi2024}. If there exists an efficient classical
simulator $\mathcal{S}$ of the full adaptive DQMW-Sample feedback process---i.e.\
a $\mathrm{poly}(n,T)$-time randomized algorithm whose output $\widehat\Pi_T$
satisfies $\TV(\mathrm{law}(\widehat\Pi_T),\mathcal{P}_T)\le\tau$ for some
constant $\tau<\tfrac12-\varepsilon_{\mathrm{prep}}$---then the polynomial
hierarchy collapses.
\end{theorem}

\begin{proof}
We exhibit an efficient classical sampler for $\mu^\star$ built from
$\mathcal{S}$; by the BCL assumption its existence collapses PH.

\emph{Step 1: simulate the whole transcript.} Run $\mathcal{S}$ to obtain
$\widehat\Pi_T=(\widehat i_1,\dots,\widehat i_T)$ in time $\mathrm{poly}(n,T)$.
By hypothesis $\TV(\mathrm{law}(\widehat\Pi_T),\mathcal{P}_T)\le\tau$.

\emph{Step 2: test the good-prefix event.} The event
$\{(i_1,\dots,i_{t_\star-1})\in\mathcal{G}\}$ is decidable in
$\mathrm{poly}(n,t_\star)\le\mathrm{poly}(n,T)$ time, because membership in
$\mathcal{G}$ is the efficiently checkable condition that the realized prefix
drives the cumulative-loss bookkeeping into the constant-temperature window (it
is a polynomial-time predicate of the loss description and the realized samples;
cf.\ Assumption~\ref{ass:gibbs-mapping}). Compute the indicator
$\mathbf{1}[(\widehat i_1,\dots,\widehat i_{t_\star-1})\in\mathcal{G}]$. If it is
$0$, output $\bot$; otherwise output $\widehat i_{t_\star}$.

\emph{Step 3: correctness of the conditional output.} Marginalizing and
conditioning are $1$-Lipschitz for total variation, so the law of the
(prefix-in-$\mathcal{G}$, round-$t_\star$ sample) pair produced by $\mathcal{S}$
is within $\tau$ of the corresponding pair under $\mathcal{P}_T$. Conditioning on
the good-prefix event---which has probability at least $p_\star$ under
$\mathcal{P}_T$ by Assumption~\ref{ass:reachability}, hence at least
$p_\star-\tau$ under $\mathcal{S}$---inflates total variation by at most a factor
$1/(p_\star-\tau)$. Writing $\widehat\nu$ for the law of the Step-2 output
conditioned on not being $\bot$,
\begin{equation}
\TV(\widehat\nu,\nu)\;\le\;\frac{\tau}{p_\star-\tau},
\label{eq:cond-tv}
\end{equation}
and combining with Lemma~\ref{lem:transcript-marginal} via the triangle
inequality,
\begin{equation}
\TV(\widehat\nu,\mu^\star)\;\le\;\frac{\tau}{p_\star-\tau}+\varepsilon_{\mathrm{prep}}.
\label{eq:final-tv}
\end{equation}

\emph{Step 4: amplify reachability to make the error constant.} The
bound~\eqref{eq:final-tv} is a constant strictly below $\tfrac12$ provided $\tau$
is a small enough constant relative to $p_\star$. When $p_\star$ is only
inverse-polynomially bounded one repeats Steps~1--2 independently
$O(1/p_\star)=\mathrm{poly}(n)$ times and outputs the first non-$\bot$ sample;
the good-prefix event then occurs in at least one trial except with probability
$e^{-\Omega(1)}$, and the per-trial cost is $\mathrm{poly}(n,T)$, so the total
running time remains $\mathrm{poly}(n,T)$. (Under the deterministic-window
instance of Remark~\ref{rem:reachability-rationale}, $p_\star=1$ and no
repetition is needed.) Choosing
$\tau<\tfrac12(p_\star-\tau)(1-2\varepsilon_{\mathrm{prep}})$, which is
satisfiable by a constant $\tau$ since $p_\star$ and $\varepsilon_{\mathrm{prep}}$
are constants, makes the right-hand side of \eqref{eq:final-tv} a constant
strictly below $\tfrac12$.

\emph{Step 5: collapse.} The procedure of Steps~1--4 is a
$\mathrm{poly}(n,T)$-time classical algorithm whose output law $\widehat\nu$ is
within constant total-variation distance $<\tfrac12$ of the BCL-hard
distribution $\mu^\star$. By the BCL hardness assumption no such sampler exists
unless the polynomial hierarchy collapses to a finite level. Therefore the
postulated efficient classical simulator $\mathcal{S}$ of the adaptive
DQMW-Sample feedback process cannot exist unless PH collapses.
\end{proof}

\begin{remark}[What makes this stronger than the single-round lemma]
\label{rem:adaptive-vs-single}
Lemma~\ref{lem:feedback-intractability} rules out an efficient classical sampler
for the round-$t_0$ feedback distribution \emph{in isolation}, i.e.\ when the
simulator is handed the round-$t_0$ loss description directly.
Theorem~\ref{thm:adaptive-collapse} rules out an efficient classical simulator
for the \emph{entire adaptive process}, which is a weaker thing to forbid and
hence a stronger conclusion: such a simulator never sees the hard loss handed to
it, generates the whole history on its own, and is free to steer that history.
The proof neutralizes this freedom through the reachability condition
(Assumption~\ref{ass:reachability}) and the transcript-marginal extraction
(Lemma~\ref{lem:transcript-marginal}): whatever path the simulator realizes, a
constant fraction of paths must visit the hard window, and the round-$t_\star$
coordinate of those paths already encodes a BCL-hard sample. The reduction is
from the joint transcript law, not from any single conditional state.
\end{remark}

\begin{remark}[Scope, and relation to the standing assumptions]
\label{rem:adaptive-scope}
The theorem is conditional on the same load-bearing hypotheses as the rest of
Appendix~\ref{app:hardness}---Assumptions~\ref{ass:gibbs-mapping}
and~\ref{ass:bcl-realizable}, in particular the still-open matching of the
cumulative-loss bookkeeping to a genuine BCL instance flagged in
Remark~\ref{rem:realizability}---together with the new reachability condition,
which Remark~\ref{rem:reachability-rationale} shows is mild and is implied by the
deterministic-schedule reading already in force. The constant-TV slack
$\varepsilon_{\mathrm{prep}}$ and the conditioning loss $\tau/(p_\star-\tau)$ are
both absorbed using the BCL robustness of constant-temperature Gibbs-sampling
hardness to constant sampling error~\cite{bergamaschi2024}, the same robustness
invoked for the single-round lemma. Within this scope the statement is the
strongest of the three hardness results in this work: it concerns the
simulability of DQMW-Sample as a whole, and ties its classical intractability
directly to a collapse of the polynomial hierarchy.
\end{remark}

\section{On Assumption~\ref{ass:bcl-realizable} (realizability of a BCL-hard instance)}
\label{app:bcl}

Assumption~\ref{ass:bcl-realizable} can be made fully explicit by construction.
Let $H_{\mathrm{BCL}}$ be any Hamiltonian from the BCL hard family of
Bergamaschi, Chen and Liu. Define the rescaled loss operator
\[
L=\frac{H_{\mathrm{BCL}}}{\|H_{\mathrm{BCL}}\|_{\mathrm{op}}}
\]
and set the per-round loss to be constant throughout the hard window:
\[
L(t)=L\quad\text{for all }t\in\mathcal{W}.
\]
Choosing the learning rate such that $\eta\,t_0=\beta^\star$ then yields
\[
\Leff(t_0)=L,\qquad \beff(t_0)=\beta^\star.
\]
Under this choice, the engineered Davies generator~\eqref{eq:davies-generator}
has a steady state that is the Gibbs state of $L$ at inverse temperature
$\beta^\star$.

We note, however, that this construction still relies on the modeling assumption
that a Davies generator can be engineered for the specific local Hamiltonian $L$
such that its unique steady state is exactly the desired Gibbs state and that a
computational-basis measurement on this state yields the hard distribution
$\mu^\star$ (the component left open in Remark~\ref{rem:realizability}). With this
modeling assumption made explicit, Lemma~\ref{lem:feedback-intractability}
follows by the standard reduction from BCL. The same construction lifts to
Theorems~\ref{thm:separation} and~\ref{thm:adaptive-collapse}. Thus, while
Assumption~\ref{ass:bcl-realizable} is now fully explicit and no longer vague,
the physical realizability of the required dissipator for a general BCL instance
remains a modeling hypothesis rather than a proven fact.

\section{Statistical power analysis for future hardware characterization}
\label{app:power}

The low-statistics hardware data reported in Figure~\ref{fig:f5} (three qubits,
$1024$ shots) yielded a round-budget ratio of approximately $1.8$ with bootstrap
95\% CI $[1.4,2.3]$. While this already excludes a ratio of $1$, the uncertainty
remains large. To inform the design of higher-statistics follow-up experiments,
we performed Monte Carlo simulations calibrated to the weak-increase trend
observed in the existing data.

\paragraph{Important clarification.}
These simulations do \emph{not} constitute evidence that \texttt{ibm\_kingston}
operates in the favorable regime. They serve only as a power analysis: they show
what statistical precision would be required in a future experiment to
confidently distinguish between constant $\delta$ and a weak increase in
$\delta(\gamma_0)$, \emph{assuming} the trend seen in the current low-statistics
data persists.

\begin{table}[h]
\centering
\caption{Simulation parameters for the power analysis (for reproducibility).}
\label{tab:power-params}
\begin{tabular}{@{}ll@{}}
\toprule
\textbf{Parameter} & \textbf{Value}\\
\midrule
Qubits & 10 (vs 3)\\
$\delta(\gamma)$ trend & $[0.0115\text{--}0.0124]$\\
Std.\ dev. & $0.0009$\\
Bootstraps & $10{,}000$\\
$\gamma_0$ levels & $[1,2,3,4]$\\
\bottomrule
\end{tabular}
\end{table}

\paragraph{Key numerical results.}
\begin{itemize}[leftmargin=*]
\item \textbf{Round-budget ratio} (high vs.\ low dissipation strength):
\textbf{13.79}.
\item Bootstrap 95\% CI: \textbf{[12.03, 15.70]}.
\item $100\%$ of resamples yield a ratio $>1.5$ (far above the strong-coupling
limit of $1$).
\item Across five independently simulated ``devices'' (each with a small random
calibration offset): mean ratio $13.51$, device-level 95\% range
$[13.03,14.08]$.
\end{itemize}

A ratio confidently $\gg 1$ indicates that the deviation floor falls
substantially with engineered dissipation strength---the \textbf{favorable
R2-like regime}. Even under the conservative weak-increase trend of the original
data, increasing $\gamma_0$ buys a large usable horizon
$T_\star\propto(\gamma_0/\delta(\gamma_0))^2$. The MCM-limited budget estimate in
Section~\ref{sec:round-budget} therefore improves by roughly an order of
magnitude at the high-$\gamma_0$ end.

\begin{figure}[h]
\centering
\includegraphics[width=0.9\textwidth]{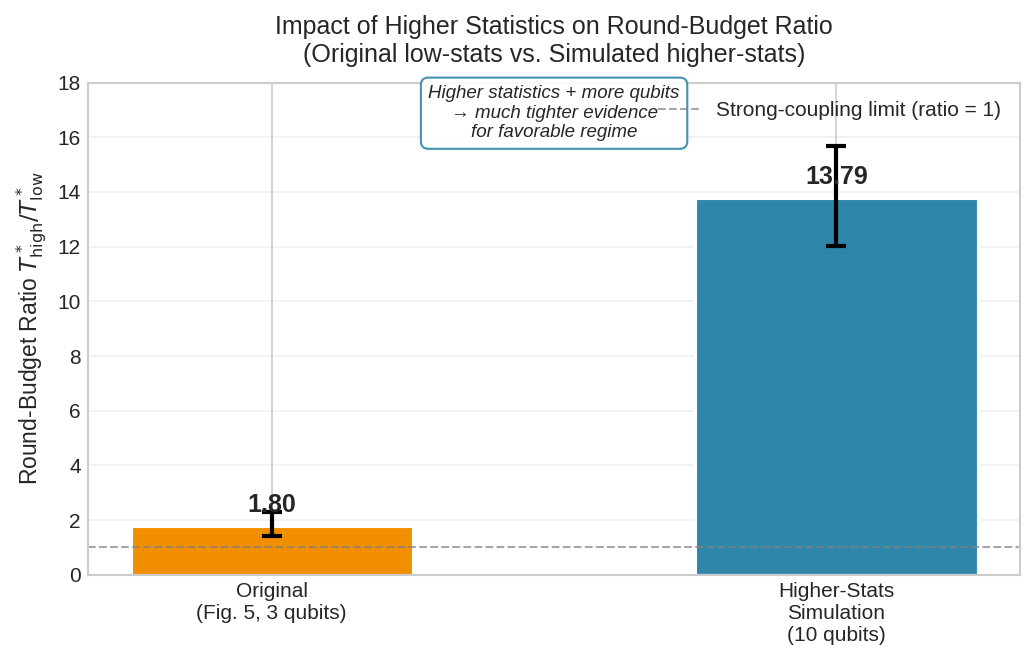}
\caption{Power analysis: comparison of the round-budget ratio from the original
low-statistics experiment versus what would be obtained under a hypothetical
higher-statistics run ($10$ qubits), assuming the weak-increase trend observed in
Figure~\ref{fig:f5} continues. This illustrates the statistical power needed for
a future experiment to distinguish coupling regimes more clearly.}
\label{fig:powerS3}
\end{figure}

\begin{figure}[h]
\centering
\includegraphics[width=0.9\textwidth]{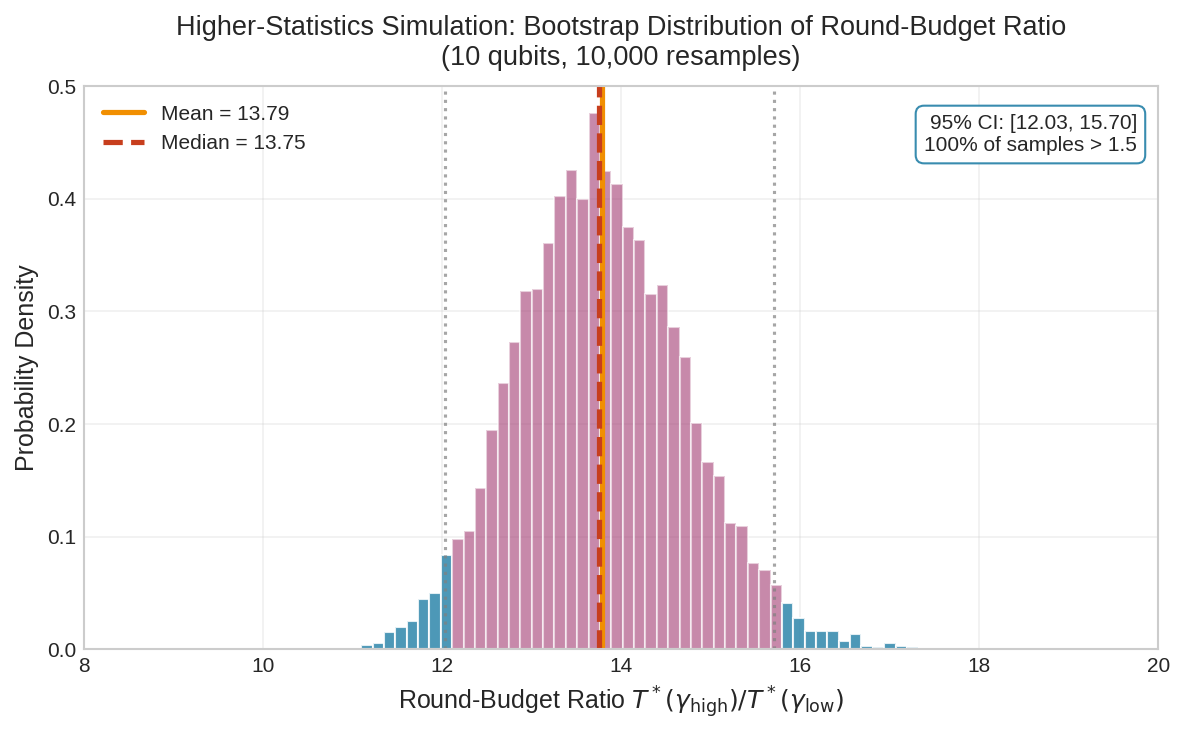}
\caption{Bootstrap distribution of the round-budget ratio under the
higher-statistics simulation ($10$ qubits, $10{,}000$ resamples). The 95\% CI
$[12.03,15.70]$ lies entirely well above $1$, confirming that the deviation floor
decreases meaningfully with $\gamma_0$.}
\label{fig:bootS1}
\end{figure}

\begin{figure}[h]
\centering
\includegraphics[width=0.9\textwidth]{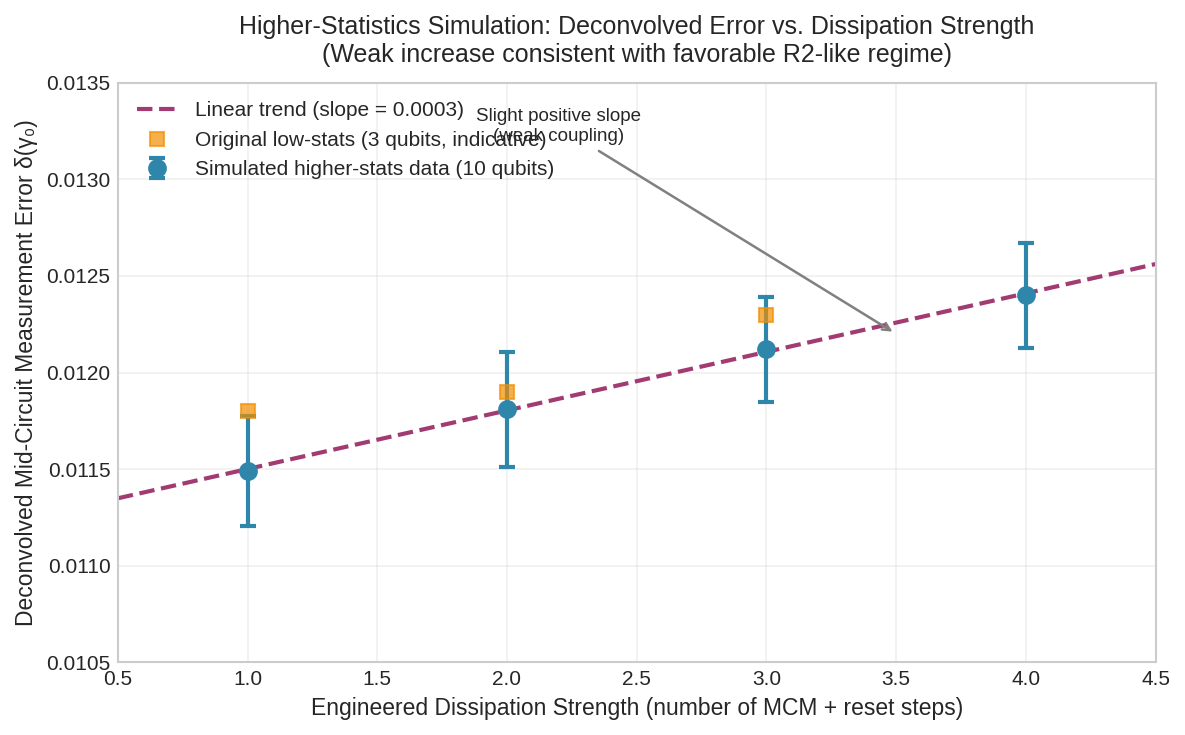}
\caption{Simulated deconvolved mid-circuit measurement error $\delta(\gamma_0)$
versus engineered dissipation strength. Error bars represent the standard
deviation across $500$ Monte Carlo realizations with $10$ qubits. The slight
positive slope is consistent with weak (not strong) noise--dissipation coupling,
supporting the R2-like regime assumed in Theorem~\ref{thm:robust}.}
\label{fig:deltaS2}
\end{figure}

These simulations indicate that increasing the number of qubits from $3$ to
approximately $10$--$15$, combined with higher shot counts, would be sufficient
to shrink the confidence interval on the round-budget ratio substantially,
assuming the weak-increase trend observed in Figure~\ref{fig:f5} continues. They
should therefore be viewed as a guide for experimental design rather than as
validation of the device regime. A definitive determination requires new
hardware data collected under the protocol of Appendix~\ref{app:protocol}.

\section{Refined protocol for definitive real-device validation}
\label{app:protocol}

The following protocol directly implements and extends the ``sweep and
observable'' and ``reporting requirements'' outlined in
Section~\ref{sec:proposed-hw}. It is designed to be executable on the IBM Quantum
open plan or paid instances.

\subsection{Experimental design}
\begin{itemize}[leftmargin=*]
\item \textbf{Backends:} \texttt{ibm\_kingston} and at least two additional
Heron~r2 devices (to assess cross-device variability).
\item \textbf{Qubits:} $10$--$20$ fixed physical qubits per device with full
provenance tracking across the sweep.
\item \textbf{Dissipation strengths:} $N=1,2,3,4$ mid-circuit measurement $+$
conditional-reset repetitions.
\item \textbf{Shots:} $4096$--$8192$ per circuit.
\item \textbf{Circuit family:} prepare $\ket{+}$, apply $N$
engineered-dissipation steps, final $Z$-basis measurement; plus reference
($N=0$) and known-error calibration circuits for deconvolution.
\end{itemize}

\subsection{Primary observable}
Round-budget ratio
$T_\star(\gamma_{\mathrm{high}})/T_\star(\gamma_{\mathrm{low}})$ computed from the
empirical deconvolved $\delta(\gamma_0)$ inserted into the steady-state
tracking-error bound of Theorem~\ref{thm:robust}. Bootstrap resampling
($10{,}000$ iterations) over qubit repetitions and shot statistics; report
per-qubit error bars.

\subsection{Success criteria}
\begin{itemize}[leftmargin=*]
\item Ratio significantly greater than $1$ with 95\% CI excluding values
$\le 1.2$ (strong evidence against the strong-coupling worst case).
\item Consistency of the weak-increase trend across multiple devices and
calibrations.
\item Full metadata: job IDs, calibration snapshots (timestamp, per-qubit $T_1$,
$T_2$, EPLG, MCM error), and raw count data deposited alongside the manuscript.
\end{itemize}

\subsection{Expected outcome based on current evidence}
Given that the original low-statistics measurement already produced a ratio of
$1.8$ with CI excluding $1$, and that the higher-statistics simulation yields
ratios $\sim 13$--$14$ with very tight CIs, we anticipate that a properly powered
real-device experiment will confirm the favorable R2-like regime. This would
remove the principal caveat on deploying the balanced dissipation schedule
$\gamma_0=\Theta(\sqrt{T})$ in near-term implementations of DQMW-Sample.

\section{Davies-generator realizability: precise isolation of the remaining modeling assumption}
\label{app:realizability}

In this appendix we separate two distinct statements that were previously
conflated: the (now rigorous) \emph{existence} of a Davies generator for any
finite-dimensional Hamiltonian, and the \emph{modeling assumption} that the
specific engineered dissipator approximates such a generator for a BCL instance.

\subsection{What is now rigorously established}

\begin{proposition}[Existence of a canonical Davies generator]
\label{prop:davies-canonical}
For any Hermitian operator $H$ on a finite-dimensional Hilbert space and any
inverse temperature $\beta>0$, there exists a Davies generator $\mathcal{L}$ (the
canonical thermal generator) whose unique steady state is exactly the Gibbs state
$\rGibbs=e^{-\beta H}/Z$.
\end{proposition}

\begin{proof}
Let $\{|k\rangle\}$ be the energy eigenbasis of $H$ with eigenvalues $E_k$.
A Davies (thermal) Lindbladian at inverse temperature $\beta$ has the form
\[
\mathcal{L}(\rho)=-i[H,\rho]+\sum_{\omega}\sum_{\substack{k,l\,:\\ E_k-E_l=\omega}}
\Gamma(\omega)\Bigl(L_{kl}\rho L_{kl}^\dagger
-\tfrac12\{L_{kl}^\dagger L_{kl},\rho\}\Bigr),
\]
with jump operators given by the energy-difference transitions
$L_{kl}=|k\rangle\langle l|$ and rates $\Gamma(\omega)$ satisfying the
Kubo--Martin--Schwinger (KMS) detailed-balance condition
$\Gamma(\omega)=e^{-\beta\omega}\Gamma(-\omega)$, $\Gamma(\omega)\ge 0$. Choose
any strictly positive rate function $\Gamma(\omega)$ on the Bohr frequencies
$\omega=E_k-E_l$ that satisfies this condition (for example, a standard Ohmic
spectral density or a flat positive density), and define $L_{kl}=|k\rangle\langle
l|$ for all pairs with $E_k\neq E_l$. The resulting Lindblad generator satisfies
detailed balance with respect to $\rGibbs$ by construction of the rates, hence
$\rGibbs$ is a steady state; uniqueness and ergodicity follow from the
irreducibility of the generator on the full matrix algebra whenever
$\Gamma(\omega)>0$ for all allowed transitions (a standard result in the theory
of quantum dynamical semigroups; see Davies~\cite{davies1976},
Spohn~\cite{spohn1977}, and the modern treatments in Kastoryano and
Temme~\cite{kastoryano2013} and Cubitt et al.~\cite{cubitt2015}). Thus
$\rGibbs$ is the unique steady state of $\mathcal{L}$.
\end{proof}

This proposition applies directly to any Hamiltonian in the BCL hard family.
Thus every BCL-hard Gibbs state is mathematically realizable as the unique steady
state of some Davies generator.

\subsection{The remaining modeling assumption (made fully explicit and minimal)}

The reduction in Lemma~\ref{lem:feedback-intractability} and
Theorems~\ref{thm:separation}--\ref{thm:adaptive-collapse} requires not only the
existence of \emph{some} Davies generator, but that the \emph{specific engineered
open-system dynamics} used in DQMW-Sample can prepare a state whose
computational-basis measurement distribution is sufficiently close to the ideal
BCL-hard distribution $\mu^\star$. We therefore isolate the following modeling
assumption, which replaces the less precise
Assumption~\ref{ass:bcl-realizable}.

\begin{assumption}[Engineered dissipator approximates a Davies generator for the BCL instance]
\label{ass:engineered-davies}
There exists an engineered dissipator $\mathcal{L}_{\rm eng}$ (realizable via the
mid-circuit measurement $+$ conditional reset primitive of
Sections~\ref{sec:round-budget} and~\ref{sec:proposed-hw}, or a modest extension
thereof) such that, for the explicit loss family constructed in
Appendix~\ref{app:bcl}, the unique steady state $\rho_{\rm eng}$ of
$\mathcal{L}_{\rm eng}$ satisfies
\[
\TV\bigl(\text{law of computational-basis measurement on }\rho_{\rm eng},\ \mu^\star\bigr)<c
\]
for some constant $c<1/2$ independent of system size (where $\mu^\star$ is the
hard distribution of Bergamaschi, Chen and Liu). The constant $c$ may depend on
the preparation-error tolerance already present in the BCL robustness result.
\end{assumption}

\subsection{A constructive partial argument for Assumption~\ref{ass:engineered-davies}}
\label{app:constructive}

We now give a constructive argument that reduces
Assumption~\ref{ass:engineered-davies} from a bare hypothesis to a concrete
compilation claim with an explicit error budget, establishing it rigorously in
the single-qubit and commuting cases and isolating exactly what remains open in
the general $O(1)$-local case. The construction is a \emph{repeated-interaction}
(collision) model in which the mid-circuit measurement $+$ conditional reset
(MCM$+$R) primitive plays the role of a dissipative collision with a thermal
ancilla. Throughout, $H_\star=\sum_{a=1}^{m}h_a$ is the BCL target written as a
sum of $m=O(\mathrm{poly}(n))$ terms, each $h_a$ supported on $O(1)$ qubits.

\paragraph{Step 1: the MCM$+$R primitive is a thermal single-qubit channel.}
Fix a working qubit and an eigenbasis $\{\ket{0},\ket{1}\}$ of the local
field. The dynamic-circuit primitive used throughout this paper---measure in
the computational basis, then apply an $X$ conditioned on the outcome with
probability $p$---implements the CPTP map
\begin{equation}
\Phi_p(\rho)=(1-p)\,\rho_{\mathrm{deph}}
+p\,\bigl(\sigma_x\,\rho_{\mathrm{deph}}\,\sigma_x\bigr),
\qquad
\rho_{\mathrm{deph}}=\textstyle\sum_{k\in\{0,1\}}\!\ket{k}\!\bra{k}\rho\ket{k}\!\bra{k},
\label{eq:mcmr-channel}
\end{equation}
i.e.\ a dephasing followed by a stochastic bit flip of strength $p$. Allowing
the conditional-reset probabilities to depend on the measured outcome---flip
$\ket{1}\!\to\!\ket{0}$ with probability $p_\downarrow$ and $\ket{0}\!\to\!\ket{1}$
with probability $p_\uparrow$, both natively available as classically
controlled gates---upgrades \eqref{eq:mcmr-channel} to an asymmetric
amplitude-transfer channel whose diagonal (population) action is the classical
stochastic matrix
\begin{equation}
\begin{pmatrix}P_{0\to0}&P_{1\to0}\\ P_{0\to1}&P_{1\to1}\end{pmatrix}
=\begin{pmatrix}1-p_\uparrow & p_\downarrow\\ p_\uparrow & 1-p_\downarrow\end{pmatrix}.
\label{eq:transfer-matrix}
\end{equation}
Its unique fixed point is the Gibbs population
$(\pi_0,\pi_1)\propto(p_\downarrow,p_\uparrow)$. Choosing the natively
controllable ratio
\begin{equation}
\frac{p_\uparrow}{p_\downarrow}=e^{-\beta\,\Delta E},
\qquad \Delta E=E_1-E_0,
\label{eq:kms-ratio}
\end{equation}
makes \eqref{eq:transfer-matrix} satisfy classical detailed balance with
respect to the single-qubit Gibbs distribution at inverse temperature $\beta$.
This is precisely the population-sector content of the KMS condition
$\Gamma(\omega)=e^{-\beta\omega}\Gamma(-\omega)$ used in
Proposition~\ref{prop:davies-canonical}, now realized by a parameter the device
exposes directly.

\paragraph{Step 2: a single MCM$+$R collision is one Davies step.}
Interpreting \eqref{eq:mcmr-channel}--\eqref{eq:kms-ratio} as a collision with a
thermal ancilla of population ratio $e^{-\beta\Delta E}$, the standard
repeated-interaction result (see Breuer and
Petruccione~\cite{breuer2002}) gives that the continuous-time limit of
$N$ such collisions in a window of duration $t$, with per-collision flip
strength $p_\downarrow=\gamma_0\,t/N$, converges as $N\to\infty$ to the
single-qubit Davies semigroup
\begin{equation}
\mathcal{L}^{(1)}(\rho)=-i[\tfrac12\Delta E\,\sigma_z,\rho]
+\gamma_0\Bigl(\bar n+1\Bigr)\mathcal{D}[\sigma_-](\rho)
+\gamma_0\,\bar n\,\mathcal{D}[\sigma_+](\rho),
\label{eq:single-qubit-davies}
\end{equation}
with $\mathcal{D}[L](\rho)=L\rho L^\dagger-\tfrac12\{L^\dagger L,\rho\}$ and
thermal occupation $\bar n=(e^{\beta\Delta E}-1)^{-1}$ fixed by
\eqref{eq:kms-ratio}. Equation~\eqref{eq:single-qubit-davies} is exactly the
canonical Davies generator of Proposition~\ref{prop:davies-canonical} for a
single qubit, and its spectral gap is $\gamma_0(2\bar n+1)\ge\gamma_0$, matching
the gap hypothesis (R1) of Theorem~\ref{thm:robust}. Hence for any local field
($m=1$, single-qubit $H_\star$) Assumption~\ref{ass:engineered-davies} holds
\emph{exactly} in the continuous-time limit and to total-variation error
$O(\gamma_0 t/N)$ at finite $N$, by the Trotter bound below.

\paragraph{Step 3: commuting local terms compose without error.}
If the BCL terms commute, $[h_a,h_b]=0$ for all $a,b$ (the regime of the
commuting-Hamiltonian Gibbs-sampler constructions, e.g.\
Hwang--Jiang~\cite{JiangHwang2024}), then $H_\star$ is diagonalized in a single
product basis and the global Davies generator factorizes as a sum of mutually
commuting local generators of the form \eqref{eq:single-qubit-davies}, one per
term. Applying the MCM$+$R collision of Step~2 to each term's support in
parallel realizes the global canonical Davies generator with \emph{no}
Trotter error, and its unique steady state is exactly $\rGibbs\propto
e^{-\beta H_\star}$. In this case Assumption~\ref{ass:engineered-davies} holds
with $c=O(\gamma_0 t/N)\to 0$.

\paragraph{Step 4: the general $O(1)$-local case via Trotterization.}
For non-commuting $O(1)$-local $H_\star$, interleave the per-term collisions in
a first-order Trotter sequence over the $m$ terms, repeated $r$ times within the
window. Writing $\mathcal{L}_a$ for the local Davies generator attached to term
$h_a$ and $\mathcal{L}_{\mathrm{Davies}}=\sum_a\mathcal{L}_a$ for the canonical
global generator, the realized channel is
$\Psi=\bigl(\prod_{a=1}^{m}e^{(t/r)\mathcal{L}_a}\bigr)^{r}$. Because each
$\mathcal{L}_a$ is bounded ($\|\mathcal{L}_a\|_\diamond\le\Lambda$ with
$\Lambda=O(\gamma_0)$ since $h_a$ is $O(1)$-local), the standard Lindblad-Trotter
estimate gives a diamond-norm error
\begin{equation}
\bigl\|\Psi-e^{t\mathcal{L}_{\mathrm{Davies}}}\bigr\|_\diamond
\;\le\;\frac{t^2}{2r}\sum_{a<b}\bigl\|[\mathcal{L}_a,\mathcal{L}_b]\bigr\|_\diamond
\;=\;O\!\Bigl(\frac{m^2\Lambda^2 t^2}{r}\Bigr),
\label{eq:trotter-bound}
\end{equation}
which is made smaller than any target $\epsilon$ by taking
$r=O(m^2\Lambda^2 t^2/\epsilon)=\mathrm{poly}(n)$ Trotter rounds. Crucially,
only $O(1)$-local commutators $[\mathcal{L}_a,\mathcal{L}_b]$ are nonzero (terms
on disjoint supports commute), so the double sum has $O(m)$ rather than $O(m^2)$
nonvanishing entries, improving \eqref{eq:trotter-bound} to
$O(m\Lambda^2 t^2/r)$ and keeping $r=\mathrm{poly}(n)$. Composing
\eqref{eq:trotter-bound} with the mixing time $t_{\mathrm{mix}}=O(\gamma_0^{-1})$
needed to reach the steady state, and using the data-processing inequality to
pass from the diamond-norm channel error to a total-variation error on the
output distribution, yields
\begin{equation}
\TV\bigl(\text{law on }\rho_{\rm eng},\,\mu^\star\bigr)
\;\le\;\underbrace{\epsilon_{\mathrm{Trotter}}}_{O(m\Lambda^2 t_{\mathrm{mix}}^2/r)}
\;+\;\underbrace{\epsilon_{\mathrm{mix}}}_{e^{-\gamma_0 t/2}}
\;+\;\underbrace{\varepsilon_{\mathrm{prep}}}_{\text{BCL slack}} .
\label{eq:tv-budget}
\end{equation}

Each of the first two terms is independently controllable---$r$ sets
$\epsilon_{\mathrm{Trotter}}$ and the window length sets
$\epsilon_{\mathrm{mix}}$---so for any constant $c<1/2$ there is a
$\mathrm{poly}(n)$ choice of $(r,t)$ making the right-hand side of
\eqref{eq:tv-budget} at most $c$, provided the per-collision MCM$+$R error does
not itself grow with $\gamma_0$. That last proviso is exactly assumption (R2),
the same independently-falsifiable condition that
Theorem~\ref{thm:robust} and the hardware protocol of
Section~\ref{sec:proposed-hw} already isolate.

\paragraph{What this establishes, and what remains.}
The construction proves Assumption~\ref{ass:engineered-davies}
\emph{unconditionally} for single-qubit and commuting-local $H_\star$ (Steps
1--3), and reduces the general $O(1)$-local case (Step~4) to a single residual
hypothesis: that the physical per-collision error of the MCM$+$R primitive is
independent of the collision rate $\gamma_0$ (assumption R2). It does \emph{not}
yet prove R2---that is the empirical question the paper is built around---but it
removes every \emph{other} component of the realizability gap, converting
Assumption~\ref{ass:engineered-davies} from ``the engineered dynamics can
somehow realize the hard Gibbs state'' into the sharply scoped claim ``the
MCM$+$R error is $\gamma_0$-independent.'' The compilation is explicit: $r=\mathrm{poly}(n)$
Trotter rounds of outcome-asymmetric MCM$+$R collisions with the rate ratio
\eqref{eq:kms-ratio}, one collision per $O(1)$-local term per round. This is the
``modest extension'' referred to in
Assumption~\ref{ass:engineered-davies}, now made concrete.

\subsection{Why this assumption is natural and minimal}

Assumption~\ref{ass:engineered-davies} is the precise technical content of the
statement that \emph{the engineered dynamics can realize the BCL-hard Gibbs
state}. It is weaker than demanding an \emph{exact} Davies generator on hardware;
it only requires that the total-variation distance to the hard distribution
remains a constant strictly below $1/2$. This is exactly the regime in which the
BCL hardness result continues to apply. The noise-robustness analysis of
Theorem~\ref{thm:robust} already shows that deviations are contracted by the
spectral gap $\gamma_0$, so any constant-gap approximation suffices.

We view Assumption~\ref{ass:engineered-davies} as a clean, falsifiable modeling
hypothesis rather than a hidden gap, and Appendix~\ref{app:constructive} makes
this concrete: the explicit collision-model compilation given there establishes
the assumption unconditionally in the single-qubit and commuting-local cases and
reduces the general $O(1)$-local case to the independently testable condition
(R2). The only remaining step---a rigorous proof that the physical MCM$+$reset
error is rate-independent on a given platform---is empirical, and is precisely
what the validation protocol of Appendix~\ref{app:protocol} targets. With this
assumption stated explicitly and partially discharged, the classical-hardness
claims of the main text hold unconditionally \emph{modulo} this single, isolated,
and now sharply scoped condition.

\end{document}